\documentclass[11pt,draftcls,onecolumn]{IEEEtran}
\usepackage{amsmath}
\usepackage{color}
\usepackage{cases}
\usepackage{amsfonts,amsmath,amssymb}
\usepackage{mathrsfs}
\usepackage{cite}
\usepackage{graphicx}
\usepackage{url}
\usepackage{enumerate}
\usepackage{subfigure}
\usepackage{hyperref}
\usepackage{algorithm}
\usepackage{algpseudocode}
\usepackage{amsmath,bm}
\usepackage{caption}
\usepackage{float}
\usepackage{multirow}
\renewcommand{\tablename}{\bfseries{Table}}
\renewcommand\thetable{\normalfont\arabic{table}}

\newtheorem{theorem}{Theorem}
\newtheorem{proposition}{Proposition}
\newtheorem{lemma}{Lemma}
\newtheorem{corollary}{Corollary}
\newtheorem{example}{Example}
\newtheorem{remark}{Remark}

\begin{document}

\title{Towards Robustness in Residue Number Systems}

\author{\mbox{Li Xiao, Xiang-Gen Xia, and Haiye Huo}
\thanks{L. Xiao and X.-G. Xia are with the Department of
Electrical and Computer Engineering,
University of Delaware, Newark, DE 19716, U.S.A. (e-mail:
\{lixiao, xxia\}@ee.udel.edu).}

\thanks{H. Huo is with the School of Mathematical Sciences and LPMC, Nankai University, Tianjin 300071, China. (e-mail:
hyhuo@mail.nankai.edu.cn).}
}
\maketitle

\begin{abstract}
The problem of robustly reconstructing a large number from its erroneous remainders with respect to several moduli, namely the robust remaindering problem, may occur in many applications including phase unwrapping, frequency detection from several undersampled waveforms, wireless sensor networks, etc. Assuming that the dynamic range of the large number is the maximal possible one, i.e., the least common multiple (lcm) of all the moduli, a method called robust Chinese remainder theorem (CRT) for solving the robust remaindering problem has been recently proposed. In this paper, by relaxing the assumption that the dynamic range is fixed to be the lcm of all the moduli, a trade-off between the dynamic range and the robustness bound for two-modular systems is studied. It basically says that a decrease in the dynamic range may lead to an increase of the robustness bound. We first obtain a general condition on the remainder errors and derive the exact dynamic range with a closed-form formula for the robustness to hold. We then propose simple closed-form reconstruction algorithms. Furthermore, the newly obtained two-modular results are applied to the robust reconstruction for multi-modular systems and generalized to real numbers. Finally, some simulations are carried out to verify our proposed theoretical results.
\end{abstract}
\begin{IEEEkeywords}
Chinese remainder theorem, dynamic range, frequency estimation from undersamplings, residue number systems, robust reconstruction.
\end{IEEEkeywords}
\newpage

\section{Introduction}
The Chinese remainder theorem (CRT) also known as \textit{Sunzi Theorem} provides a reconstruction formula for a large nonnegative integer from its remainders with respect to several moduli, if the large integer is less than the least common multiple (lcm) of all the moduli. The CRT has applications in many fields, such as computing, cryptograph, and digital signal processing \cite{CRT3,CRT1,CRT2}. Note that all the remainders in the CRT reconstruction formula have to be error-free, because a small error in a remainder may cause a large reconstruction error.
In this work, we consider a problem of robustly reconstructing a large nonnegative integer when the remainders have errors, called the robust remaindering problem, and its applications can be found in phase unwrapping in radar signal processing \cite{wxu1994,jorgensen2000,gwang2004,mruegg2007,gli1-2007,gli2-2008,yimin2008,yimin2015,wkqi2009,xwli1-2011,grslet2013,aa2015}, multiwavelength optical measurement \cite{fala2-2011,fala1-2013,tang}, wireless sensor networks \cite{dad3-2003,dad4-2012,deng,chessa2012,dad1-2013,yishengsu}, and computational neuroscience \cite{haft,fiete,fiete2,stemm}. In this robust remaindering problem, two fundamental questions are of interested: $1)$ What is the dynamic range of the large integer and how large can the remainder errors be for the robustness to hold? $2)$ How can the large integer be robustly reconstructed from the erroneous remainders? Here, the dynamic range is defined as the minimal positive number of the large integer such that the robustness does not hold.
For the first question, the larger the dynamic range and the remainder errors can be, the better the reconstruction is. It is not hard to see that the maximal possible dynamic range is the lcm of all the moduli. For the second question, it is the reconstruction algorithm.

When the dynamic range is assumed to be the maximal possible one, i.e., the lcm of all the moduli, a robust CRT method for solving the robust remaindering problem has been investigated in \cite{xia1-2007,xwli2-2009,wjwang2010,guangwu,binyang2014,xiaoxia1,xiaoxia2}. In these papers, the folding integers (i.e., the quotients of the large integer divided by the moduli) are first accurately determined, and a robust reconstruction is then given by the average of the reconstructions obtained from the folding integers. In \cite{xia1-2007,xwli2-2009,wjwang2010,guangwu}, a special case when the remaining integers of the moduli factorized by their greatest common divisor (gcd) are pairwise co-prime was considered. It basically says that the reconstruction error is upper bounded by the remainder error bound $\tau$ if $\tau$ is smaller than a quarter of the gcd of all the moduli (see Proposition \ref{pr2} in Section \ref{sec2}). Notably, a necessary and sufficient condition for accurate determination of the folding integers (see Proposition \ref{pr1} in Section \ref{sec2}) and their closed-form determination algorithm were presented in \cite{wjwang2010}. Recently, an improved version of robust CRT, called multi-stage robust CRT, was proposed in \cite{binyang2014,xiaoxia1}, where the remaining integers of the moduli factorized by their gcd are not necessarily pairwise co-prime. It is shown in \cite{xiaoxia1} that the remainder error bound may be above the quarter of the gcd of all the moduli. By relaxing the assumption that the dynamic range is fixed to the maximum, i.e., the lcm of all the moduli,
another method of position representation on the remainder plane was proposed for solving
the robust remaindering problem with only two moduli $m_1,m_2$ and $m_1<m_2$ in \cite{new}. Different from the robust CRT, all the nonnegative integers less than the dynamic range are connected by the slanted lines with the slope of $1$ on the two dimensional remainder plane and a robust reconstruction is obtained by finding the closest point to the erroneous remainders on one of the slanted lines in \cite{new}. As the dynamic range increases, the number of the slanted lines increases, and thereby, the distance between the slanted lines decreases, that is, the remainder error bound becomes small. In \cite{new},
an exact dynamic range was first presented, provided that the remainder error bound is smaller than a quarter of the remainder of $m_2$ modulo $m_1$ (see Proposition \ref{pr3} in Section \ref{sec2}). When the remainder of $m_2$ modulo $m_1$ does not equal the gcd of $m_1$ and $m_2$, an extension with a smaller remainder error bound and a larger dynamic range was further obtained, and as
the dynamic range increases to the lcm of the moduli, the remainder error bound will decrease to the quarter of the gcd of the moduli (see Proposition \ref{pr4} in Section \ref{sec2}). In \cite{new}, however, no closed-form reconstruction algorithms were proposed, and in the extension result, only lower and upper bounds of the dynamic range were provided, while the exact one was not derived or given.
In some practical applications, considering that an unknown is real-valued in general, the robust remaindering problem and the above two different solutions were naturally generalized to real numbers in \cite{wjwang2010,wenchao2013,wenjie2015} and \cite{new}.

Different from robustly reconstructing a large integer from its erroneous remainders in the robust remaindering problem, another technique to resist remainder errors, i.e.,
the Chinese remainder code as an error-correcting code based on Redundant Residue Number Systems, has been studied extensively in \cite{kls1-1992,kls2-1992,vgoh2008,lle2-2000,lly2-2002,lly3-2002,lle-2006,sundaram,goldreich,guruswa,boneh,shparlinski}.
When only a few of the remainders are allowed to have errors and most of the remainders have to be error-free, there has been a series of results on unique decoding of the Chinese remainder code in \cite{kls1-1992,kls2-1992,vgoh2008,lle2-2000,lly2-2002,lly3-2002,lle-2006,sundaram}, where the large integer is
accurately recovered as a unique output in the decoding algorithm. If the number of the remainder errors is larger, i.e., the error rate is larger, list decoding of the Chinese remainder code has been investigated as a generalization of unique decoding in \cite{goldreich,guruswa,boneh,shparlinski}, where the decoding algorithm outputs a small list of possibilities one of which is accurate.

In this paper, we are interested in the robust remaindering problem with only two moduli as in \cite{new} and consider the relationship between the dynamic range and the remainder errors. Motivated from the robust CRT in \cite{wjwang2010}, we want to accurately determine the folding integers from the erroneous remainders in this paper. Compared with the condition that the remainder error bound is smaller than a quarter of the remainder of $m_2$ modulo $m_1$ (see Proposition \ref{pr3} in Section \ref{sec2}), we first present a general condition on the remainder errors such that the folding integers can be accurately determined, and a simple closed-form determination algorithm is proposed in this paper. We then extend this result, if the remainder of $m_2$ modulo $m_1$ does not equal the gcd of $m_1$ and $m_2$. Compared with the corresponding result (see Proposition \ref{pr4} in Section \ref{sec2}) in \cite{new}, we give the exact dynamic range with a closed-form formula, and we also present
a general condition on the remainder errors and a closed-form algorithm for accurate determination of the folding integers. Finally, the newly obtained results are applied to multi-modular systems by using cascade architectures, and generalized to real numbers in this paper.

The rest of the paper is organized as follows. In Section \ref{sec2}, we briefly state the robust remaindering problem and review two existing different solving methods obtained in \cite{wjwang2010,new}. In Section \ref{sec3}, compared with the result (see Proposition \ref{pr3} in Section \ref{sec2}) in \cite{new}, we present a simple closed-form algorithm for accurate determination of the folding integers and derive a general condition on the remainder errors. In Section \ref{sec4}, we extend the result obtained in Section \ref{sec3}, and furthermore, the exact dynamic range is derived and a closed-form determination algorithm is also proposed. In Section \ref{sec5}, we study robust reconstruction for multi-modular systems and a generalization to real numbers based on the newly obtained results. In Section \ref{sec6}, we present some simulation results to demonstrate the performance of our proposed algorithms. In Section \ref{sec7}, we conclude the paper.

\textit{Notations:} The gcd and the lcm of two or more positive integers $a_1,a_2,\cdots,a_L$ are denoted by $\mbox{gcd}(a_1,a_2,\cdots,a_L)$ and $\mbox{lcm}(a_1,a_2,\cdots,a_L)$, respectively. Two positive integers are said to be co-prime, if their gcd is $1$. Given two positive integers $a$ and $b$, the remainder of $a$ modulo $b$ is denoted as $|a|_b$. It is well known that $\lfloor\ast\rfloor$, $\lceil\ast\rceil$, and $[\ast]$ stand for the floor, ceiling, and rounding functions, respectively. To distinguish from integers, we use boldface symbols to denote the real-valued variables.

\section{Preliminaries}\label{sec2}
Let $N$ be a nonnegative integer, $1<m_1<m_2<\cdots<m_L$ be $L$ moduli, and $r_1,r_2,\cdots,r_L$ be the corresponding remainders of $N$, i.e.,
\begin{equation}\label{fold}
r_i\equiv N\mbox{ mod }m_i\quad\mbox{ or }\quad N=n_im_i+r_i,
\end{equation}
where $0\leq r_i<m_i$, and $n_i$ is an unknown integer which is called folding integer, for $1\leq i\leq L$. It is well known that when $N$ is less than the lcm of all the moduli, $N$ can be uniquely reconstructed from its remainders via the CRT \cite{CRT1,CRT2,CRT3} as
\begin{equation}\label{reconCRT}
N=\left|\sum_{j=1}^{L}r_jD_jM_j\right|_{\mbox{lcm}(m_1,m_2,\cdots,m_L)},
\end{equation}
where $M_j=\mbox{lcm}(m_1,m_2,\cdots,m_L)/\mu_j$, $D_j$ is the modular multiplicative inverse of $M_j$ modulo $\mu_j$ (i.e.,
$1\equiv D_jM_j\mbox{ mod }\mu_j$), if $\mu_j\neq1$, else $D_j=0$, and $\{\mu_1,\mu_2,\cdots,\mu_L\}$ is a set of $L$ pairwise co-prime positive integers such that $\prod_{i=1}^{L}\mu_i=\mbox{lcm}(m_1,m_2,\cdots,m_L)$ and $\mu_i$ divides $m_i$ for each $1\leq i\leq L$. In particular, when moduli $m_i$ are pairwise co-prime, we can let $\mu_i=m_i$ for $1\leq i\leq L$, and then the above reconstruction formula in (\ref{reconCRT}) reduces to the traditional CRT with pairwise co-prime moduli.

The problem we are interested is to robustly reconstruct $N$ when the remainders $r_i$ have errors:
\begin{equation}
0\leq\tilde{r}_i<m_i\quad\mbox{ and }\quad |\tilde{r}_i-r_i|\leq\tau,
\end{equation}
where $\triangle r_i\triangleq \tilde{r}_i-r_i$ is the remainder error, and $\tau$ is an error level, also called remainder error bound. Now we want to reconstruct $N$ from the known moduli and the erroneous remainders such that the reconstruction error is linearly bounded by the remainder error bound $\tau$. This problem has two aspects. The first aspect is what the dynamic range of $N$ is and how large the remainder errors can be for the robustness to hold. Clearly, the larger the dynamic range and the remainder error bound $\tau$ are for the robustness, the better the reconstruction is.
The second aspect is the reconstruction algorithm.

In what follows, we briefly describe two different methods for solving the robust remaindering problem, respectively introduced in \cite{wjwang2010} and \cite{new}.

\subsection{Method of Robust CRT}

Suppose that the dynamic range of $N$ is the maximal possible one, i.e., the lcm of all the moduli. A robust reconstruction method, i.e., robust CRT, has been studied in
\cite{xia1-2007,xwli2-2009,wjwang2010,guangwu,binyang2014,xiaoxia1,xiaoxia2}, where the basic idea is to accurately determine the unknown folding integers $n_i$ for $i=1,2,\cdots,L$ in (\ref{fold}) that may cause large errors in the reconstruction if they are erroneous. Once the folding integers are accurately found, an estimate of $N$ can be given by
\begin{equation}\label{reconN}
\begin{split}
\hat{N}&=\left[\frac{1}{L}\sum_{i=1}^{L}(n_im_i+\tilde{r}_i)\right]\\
&=N+\left[\frac{1}{L}\sum_{i=1}^{L}\triangle r_i\right].
\end{split}
\end{equation}
Recall that $[\ast]$ denotes the rounding function, i.e., for any real number $\textbf{x}$, $[\textbf{x}]$ is an integer subject to
\begin{equation}
-\frac{1}{2}\leq \textbf{x}-[\textbf{x}]<\frac{1}{2}.
\end{equation}
In fact, $[\textbf{x}]=\lfloor\textbf{x}+0.5\rfloor$.
From $|\triangle r_i|\leq\tau$ for $1\leq i\leq L$, one can see that
\begin{equation}
|\hat{N}-N|\leq\tau,
\end{equation}
i.e., $\hat{N}$ in (\ref{reconN}) is a robust estimate of $N$.

Write $m_i=m\Gamma_i$ for $1\leq i\leq L$, where $m$ is the gcd of all the moduli, i.e., $m=\mbox{gcd}(m_1,m_2,\cdots,m_L)$.
When the remaining integers $\Gamma_i$ of the moduli factorized by the gcd are pairwise co-prime, an integer $N$ with $0\leq N<\mbox{lcm}(m_1,m_2,\cdots,m_L)$ can be robustly reconstructed with the reconstruction error upper bounded by the remainder error bound $\tau$, if $\tau$ is smaller than a quarter of the gcd of all the moduli \cite{xia1-2007,xwli2-2009,wjwang2010}. In particular, a necessary and sufficient condition for accurate determination of the folding integers and their closed-form determination algorithm were obtained in \cite{wjwang2010}.

\begin{proposition}[\cite{wjwang2010}]\label{pr1}
Let $m_i=m\Gamma_i$ for $1\leq i\leq L$ and $0\leq N<\mbox{lcm}(m_1,m_2,\cdots,m_L)$.
Assume that $\Gamma_i$ for $1\leq i\leq L$ are pairwise co-prime. Then, the folding integers $n_i$ for $1\leq i\leq L$ can be accurately determined, if and only if
\begin{equation}\label{ccc}
-\frac{1}{2}\leq\frac{\triangle r_i-\triangle r_1}{m}<\frac{1}{2}\quad\mbox{for all }2\leq i\leq L.
\end{equation}
\end{proposition}

For the closed-form determination algorithm of Proposition \ref{pr1}, we refer the reader to \cite{wjwang2010}. Moreover, with the condition (\ref{ccc}) in Proposition \ref{pr1}, the following result becomes various.
\begin{proposition}[\cite{xia1-2007,xwli2-2009,wjwang2010}]\label{pr2}
Let $m_i=m\Gamma_i$ for $1\leq i\leq L$ and $0\leq N<\mbox{lcm}(m_1,m_2,\cdots,m_L)$.
Assume that $\Gamma_i$ for $1\leq i\leq L$ are pairwise co-prime. Then, the folding integers $n_i$ for $1\leq i\leq L$ can be accurately determined, if the remainder error bound $\tau$ satisfies
\begin{equation}
|\triangle r_i|\leq\tau<\frac{m}{4}.
\end{equation}
\end{proposition}

Furthermore, considering a general set of moduli, i.e., the remaining integers of the moduli factorized by their gcd are not necessarily pairwise co-prime, an improved version of robust CRT, called multi-stage robust CRT, was recently proposed in \cite{binyang2014,xiaoxia1}. In \cite{xiaoxia1}, it is shown that the remainder error bound $\tau$ may be above the quarter of the gcd of all the moduli.

\subsection{Method of Integer Position Representation on The Remainder Plane}
In \cite{new}, the author considered the robust remaindering problem in the case of two moduli by using its distinctive method of position representation on the two dimensional remainder plane. Different from the robust CRT in Propositions \ref{pr1} and \ref{pr2} where the dynamic range is fixed to the maximum, i.e., the lcm of all the moduli, the relationship between the dynamic range of $N$ and the remainder error bound $\tau$ was investigated in \cite{new}, and it is shown that if the dynamic range becomes less than the lcm of the moduli, the robust reconstruction may hold even when the remainder error bound is over the quarter of the gcd of the moduli.

Let us first give an intuitive explanation with an example in the following.
Let $m_1=2\cdot12=24$ and $m_2=2\cdot19=38$. We represent the integers from $0$ to $76$ with respect to their remainders $r_1$ and $r_2$, as depicted in Fig. \ref{figone}. One can see that these integers are connected by the lines with the slope of $1$ in Fig. \ref{figone}. Let $R\triangleq(r_1,r_2)$ be an integer. When the remainders have errors, i.e., $|\tilde{r}_i-r_i|\leq\tau$ for $i=1,2$, the point $R^{\prime}\triangleq(\tilde{r}_1,\tilde{r}_2)$ may not locate on one of the slanted lines in Fig. \ref{figtwo}. The idea is to find the closest point $R^{\prime\prime}$ to $R^{\prime}$ on one of the slanted lines. It is readily seen from Fig. \ref{figtwo} that as long as the remainder error bound $\tau$ is smaller than $d/4$, the closest line to $R^{\prime}$ would be the same line that contains the true integer $R\triangleq(r_1,r_2)$. This is because the coordinates of $R^{\prime\prime}$ in Fig. \ref{figtwo} are $\left(r_1+(\triangle r_1+\triangle r_2)/2,r_2+(\triangle r_1+\triangle r_2)/2\right)$, which makes $|R^{\prime\prime}-R|\leq\tau$ hold. Then, let us see the example in Fig. \ref{figone} again. When the dynamic range is $48$, i.e., $0\leq N<48$, all the integers are connected by the three red slanted lines, and it is easy to see that the remainder error bound $\tau$ can be smaller than $14/4$. When the dynamic range is $76$, i.e., $0\leq N<76$, all the integers are connected by the three red and two blue slanted lines, and it is easy to see that the remainder error bound $\tau$ can be smaller than $10/4$. Obviously, as the dynamic range increases, the distance between the new set of slanted lines decreases, i.e., the remainder error bound becomes smaller.

\begin{figure}[H]
\begin{minipage}[t]{0.5\linewidth}
\centering
\hspace*{-2.5cm}\includegraphics[width=10in]{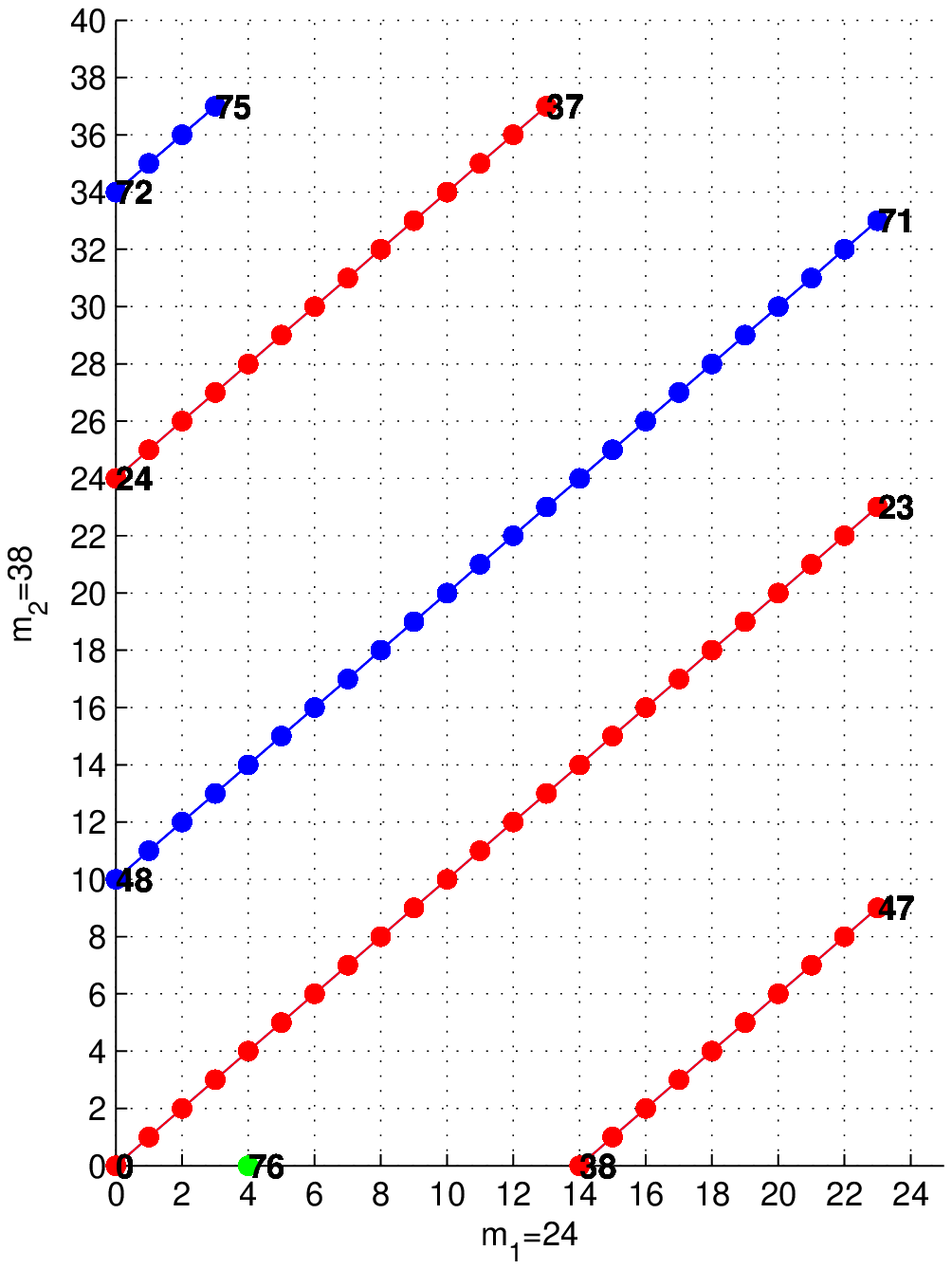}
\caption{Position representation with moduli $24$ and $38$.}
\label{figone}
\end{minipage}%
\begin{minipage}[t]{0.5\linewidth}
\hspace*{-1.3cm}\includegraphics[width=6.5in]{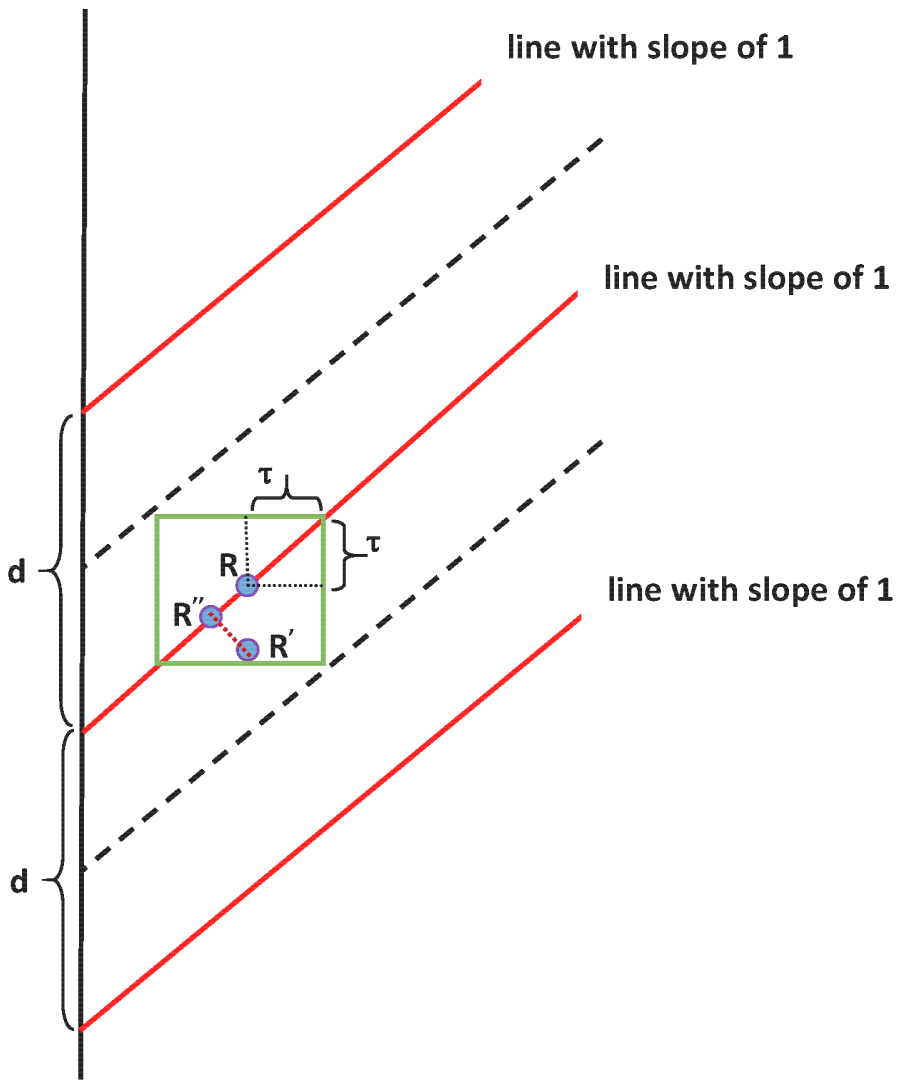}
\caption{Finding the closest point $R^{\prime\prime}$ to $R^{\prime}$ on one of the slanted lines.}
\label{figtwo}
\end{minipage}
\end{figure}

Let $m_1=m\Gamma_1,m_2=m\Gamma_2$, where $\Gamma_1$ and $\Gamma_2$ are co-prime and $1<\Gamma_1<\Gamma_2$. Let $\delta_{-1}=m_2,\delta_0=m_1,\delta_1=|m_2|_{m_1}$, and for $i\geq1$,
\begin{equation}
\delta_{i+1}=\mbox{min}(|\delta_{i-1}|_{\delta_i},\delta_i-|\delta_{i-1}|_{\delta_i}).
\end{equation}
One can see that $m$ divides each $\delta_{i}$, and there exists the largest index $G$ such that $\delta_G=m$.

\begin{proposition}[\cite{new}]\label{pr3}
If the remainder error bound $\tau$ satisfies
\begin{equation}
\tau<\frac{|m_2|_{m_1}}{4}=\frac{m|\Gamma_2|_{\Gamma_1}}{4}=\frac{\delta_{1}}{4},
\end{equation}
then the dynamic range of $N$ is $m_1\left(1+\left\lfloor m_2/m_1\right\rfloor\left\lfloor m_1/|m_2|_{m_1}\right\rfloor\right)$, i.e.,
an integer $N$ with $0\leq N<m_1\left(1+\left\lfloor m_2/m_1\right\rfloor\left\lfloor m_1/|m_2|_{m_1}\right\rfloor\right)$ can be robustly reconstructed.
\end{proposition}

Note that when $|\Gamma_2|_{\Gamma_1}=1$, i.e., the remainder of $m_2$ modulo $m_1$ equals the gcd $m$ of $m_1$ and $m_2$, we have $m_1\left(1+\left\lfloor m_2/m_1\right\rfloor\left\lfloor m_1/|m_2|_{m_1}\right\rfloor\right)=m_1\Gamma_2=\mbox{lcm}(m_1,m_2)$, and thus
Proposition \ref{pr3} coincides with Proposition \ref{pr2}. When $|\Gamma_2|_{\Gamma_1}>1$, i.e., $\delta_1>m$ or $G\geq2$, as described in Fig. \ref{figone}, a larger dynamic range associated with a smaller remainder error bound for robust reconstruction is possible, and in this case,
an extension of Proposition \ref{pr3} was also proposed in \cite{new}.

\begin{proposition}[\cite{new}]\label{pr4}
If $\delta_1>m$ and the remainder error bound $\tau$ satisfies
\begin{equation}
\tau<\frac{\delta_i}{4},\quad\mbox{for some }i, 2\leq i\leq G,
\end{equation}
then the dynamic range of $N$ is lower bounded by
$m_1\left(1+\left\lfloor m_2/m_1\right\rfloor\left\lfloor m_1/|m_2|_{m_1}\right\rfloor\right)\lfloor\delta_1/\delta_2\rfloor\cdots\lfloor\delta_{i-1}/\delta_i\rfloor$ and upper bounded by $\mbox{max}(m_1\lfloor m_2/\delta_i\rfloor,m_2\lfloor m_1/\delta_i\rfloor)$.
\end{proposition}

Note that Proposition \ref{pr4} only provides lower and upper bounds of the dynamic range, while the exact dynamic range of $N$ was not derived or given in \cite{new}. Moreover, no closed-form reconstruction algorithms for Propositions \ref{pr3} and \ref{pr4} were proposed in \cite{new}. The solution for the robust remaindering problem was naturally generalized to real numbers both in \cite{wjwang2010} and \cite{new}, and we will discuss it later in this paper.

\section{Robust Remaindering with Two Moduli}\label{sec3}
Motivated from Proposition \ref{pr1} in \cite{wjwang2010}, we first present a general condition on the remainder errors such that the folding integers of $N$ with the dynamic range given in Proposition \ref{pr3} can be accurately determined from the erroneous remainders. We then propose a simple closed-form determination algorithm to solve for the folding integers, and thus robustly reconstruct $N$ by (\ref{reconN}) in this section.

Let $m_1=m\Gamma_1,m_2=m\Gamma_2$, where $\Gamma_1$ and $\Gamma_2$ are co-prime and $1<\Gamma_1<\Gamma_2$. Before giving the result, let us introduce some necessary lemmas as follows.

\begin{lemma}\label{lem1}
Let $N$ be an integer with $0\leq N<m_1\left(1+\left\lfloor m_2/m_1\right\rfloor\left\lfloor m_1/|m_2|_{m_1}\right\rfloor\right)$ and $|\Gamma_2|_{\Gamma_1}\geq2$. Then, we have
\begin{equation}\label{nen2}
0\leq n_2\leq\left\lfloor \frac{\Gamma_1}{|\Gamma_2|_{\Gamma_1}}\right\rfloor\quad\mbox{and}\quad 0\leq n_1\leq\left\lfloor\frac{\Gamma_2}{\Gamma_1}\right\rfloor\left\lfloor \frac{\Gamma_1}{|\Gamma_2|_{\Gamma_1}}\right\rfloor.
\end{equation}
Moreover, when $n_2=\left\lfloor \Gamma_1/|\Gamma_2|_{\Gamma_1}\right\rfloor$, we have $r_1>r_2$.
\end{lemma}
\begin{proof}
From $N=n_1m_1+r_1$ and $0\leq N<m_1\left(1+\left\lfloor m_2/m_1\right\rfloor\left\lfloor m_1/|m_2|_{m_1}\right\rfloor\right)$, it is easy to see that $0\leq n_1\leq\left\lfloor m_2/m_1\right\rfloor\left\lfloor m_1/|m_2|_{m_1}\right\rfloor=\left\lfloor\Gamma_2/\Gamma_1\right\rfloor\left\lfloor \Gamma_1/|\Gamma_2|_{\Gamma_1}\right\rfloor$.
Next, according to $m_2=\left\lfloor m_2/m_1\right\rfloor m_1+|m_2|_{m_1}$,
we can equivalently write $m_1\left(1+\left\lfloor m_2/m_1\right\rfloor\left\lfloor m_1/|m_2|_{m_1}\right\rfloor\right)$ as
\begin{equation}\label{aaa}
\begin{split}
m_1\left(1+\left\lfloor\frac{m_2}{m_1}\right\rfloor\left\lfloor\frac{m_1}{|m_2|_{m_1}}\right\rfloor\right)
&=\left(\left\lfloor\frac{m_2}{m_1}\right\rfloor m_1+|m_2|_{m_1}\right)\left\lfloor \frac{m_1}{|m_2|_{m_1}}\right\rfloor+\left(m_1-|m_2|_{m_1}\left\lfloor \frac{m_1}{|m_2|_{m_1}}\right\rfloor\right)\\
&=m_2\left\lfloor\frac{m_1}{|m_2|_{m_1}}\right\rfloor+\left(m_1-|m_2|_{m_1}\left\lfloor \frac{m_1}{|m_2|_{m_1}}\right\rfloor\right).
\end{split}
\end{equation}
Also, since $\Gamma_2=\left\lfloor\Gamma_2/\Gamma_1\right\rfloor\Gamma_1+|\Gamma_2|_{\Gamma_1}$,
we can obtain that $\Gamma_1$ and $|\Gamma_2|_{\Gamma_1}$ are co-prime when $|\Gamma_2|_{\Gamma_1}\neq1$. It is due to the fact that $\Gamma_1$ and $\Gamma_2$ are co-prime.
So, we have $m_2>m_1-|m_2|_{m_1}\left\lfloor m_1/|m_2|_{m_1}\right\rfloor>0$ in (\ref{aaa}). Thus,
we have
\begin{equation}
0\leq n_2\leq\left\lfloor \frac{m_1}{|m_2|_{m_1}}\right\rfloor=\left\lfloor \frac{\Gamma_1}{|\Gamma_2|_{\Gamma_1}}\right\rfloor,
\end{equation}
when $0\leq N<m_1\left(1+\left\lfloor m_2/m_1\right\rfloor\left\lfloor m_1/|m_2|_{m_1}\right\rfloor\right)$.

Furthermore, due to $N=n_im_i+r_i$ for $i=1,2$, we get
\begin{equation}
n_2\Gamma_2-n_1\Gamma_1=\frac{r_1-r_2}{m}.
\end{equation}
When $n_2=\left\lfloor \Gamma_1/|\Gamma_2|_{\Gamma_1}\right\rfloor$, we have
\begin{equation}
\begin{split}
\frac{r_1-r_2}{m}&=\left\lfloor \frac{\Gamma_1}{|\Gamma_2|_{\Gamma_1}}\right\rfloor\Gamma_2-n_1\Gamma_1\\
&\geq \left\lfloor \frac{\Gamma_1}{|\Gamma_2|_{\Gamma_1}}\right\rfloor\Gamma_2-\left\lfloor\frac{\Gamma_2}{\Gamma_1}\right\rfloor\left\lfloor \frac{\Gamma_1}{|\Gamma_2|_{\Gamma_1}}\right\rfloor\Gamma_1\\
&=\left(\Gamma_2-\Gamma_1\left\lfloor\frac{\Gamma_2}{\Gamma_1}\right\rfloor\right)\left\lfloor \frac{\Gamma_1}{|\Gamma_2|_{\Gamma_1}}\right\rfloor>0.
\end{split}
\end{equation}
So, we obtain $r_1>r_2$ when $n_2=\left\lfloor \Gamma_1/|\Gamma_2|_{\Gamma_1}\right\rfloor$.
\end{proof}

Let
\begin{equation}
\textbf{q}_{21}\triangleq\frac{\tilde{r}_1-\tilde{r}_2}{m}.
\end{equation}
Then, we have the following result.

\begin{lemma}\label{lem3}
Let $N$ be an integer with $0\leq N<m_1\left(1+\left\lfloor m_2/m_1\right\rfloor\left\lfloor m_1/|m_2|_{m_1}\right\rfloor\right)$, $|\Gamma_2|_{\Gamma_1}\geq2$, and the remainder errors satisfy
\begin{equation}\label{con11}
-\frac{|\Gamma_2|_{\Gamma_1}}{2}\leq\frac{\triangle r_1-\triangle r_2}{m}<\frac{|\Gamma_2|_{\Gamma_1}}{2}.
\end{equation}
We can obtain the following three cases:
\begin{enumerate}
  \item if $\textbf{q}_{21}\geq|\Gamma_2|_{\Gamma_1}/2$, we have $r_1>r_2$;
  \item if $\textbf{q}_{21}<-|\Gamma_2|_{\Gamma_1}/2$, we have $r_1<r_2$;
  \item if $-|\Gamma_2|_{\Gamma_1}/2\leq \textbf{q}_{21}<|\Gamma_2|_{\Gamma_1}/2$, we have $r_1=r_2$.
\end{enumerate}
\end{lemma}
\begin{proof}
From Lemma \ref{lem1}, one can see that $\Gamma_1$ and $|\Gamma_2|_{\Gamma_1}$ are co-prime and $0\leq n_2\leq\left\lfloor \Gamma_1/|\Gamma_2|_{\Gamma_1}\right\rfloor$ if $|\Gamma_2|_{\Gamma_1}\geq2$. So, we have
\begin{equation}\label{daxiao}
0\leq n_2|\Gamma_2|_{\Gamma_1}<\Gamma_1.
\end{equation}
Then, modulo $\Gamma_1$ in both sides of $n_2\Gamma_2-n_1\Gamma_1=(r_1-r_2)/m$, we get
\begin{equation}\label{eq3}
n_2|\Gamma_2|_{\Gamma_1}\equiv\frac{r_1-r_2}{m}\mbox{ mod }\Gamma_1.
\end{equation}
When $r_1>r_2$, we have $(r_1-r_2)/m=n_2|\Gamma_2|_{\Gamma_1}$ with $n_2\geq1$ from (\ref{daxiao}) and (\ref{eq3}). Based on (\ref{con11}), we have
\begin{equation}
\frac{\tilde{r}_1-\tilde{r}_2}{m}=\frac{r_1-r_2}{m}+\frac{\triangle r_1-\triangle r_2}{m}\geq\frac{|\Gamma_2|_{\Gamma_1}}{2}.
\end{equation}
When $r_1<r_2$, we first know $0\leq n_2\leq\left\lfloor \Gamma_1/|\Gamma_2|_{\Gamma_1}\right\rfloor-1$ from Lemma \ref{lem1}. Then, from (\ref{eq3}), we get
\begin{equation}\label{e23}
\begin{split}
\frac{r_2-r_1}{m}&=k\Gamma_1-n_2|\Gamma_2|_{\Gamma_1}\mbox{ with }k\geq1\\
&\geq k\Gamma_1-\left(\left\lfloor \frac{\Gamma_1}{|\Gamma_2|_{\Gamma_1}}\right\rfloor-1\right)|\Gamma_2|_{\Gamma_1}\\
&=\left(k\Gamma_1-\left\lfloor \frac{\Gamma_1}{|\Gamma_2|_{\Gamma_1}}\right\rfloor|\Gamma_2|_{\Gamma_1}\right)+|\Gamma_2|_{\Gamma_1}\\
&>|\Gamma_2|_{\Gamma_1}.
\end{split}
\end{equation}
Based on (\ref{con11}), we have
\begin{equation}
\frac{\tilde{r}_1-\tilde{r}_2}{m}=\frac{r_1-r_2}{m}+\frac{\triangle r_1-\triangle r_2}{m}<-\frac{|\Gamma_2|_{\Gamma_1}}{2}.
\end{equation}
When $r_1=r_2$, we have $\tilde{r}_1-\tilde{r}_2=\triangle r_1-\triangle r_2$. Based on (\ref{con11}), we have
\begin{equation}
-\frac{|\Gamma_2|_{\Gamma_1}}{2}\leq\frac{\tilde{r}_1-\tilde{r}_2}{m}<\frac{|\Gamma_2|_{\Gamma_1}}{2}.
\end{equation}
Therefore, we can obtain the above three cases and complete the proof.
\end{proof}

\begin{lemma}\label{lem2}
Let $N$ be an integer with $0\leq N<m_1\left(1+\left\lfloor m_2/m_1\right\rfloor\left\lfloor m_1/|m_2|_{m_1}\right\rfloor\right)$, $|\Gamma_2|_{\Gamma_1}\geq2$, and the remainder errors satisfy (\ref{con11}).
When $\textbf{q}_{21}<-|\Gamma_2|_{\Gamma_1}/2$, we can obtain that if
\begin{equation}
\frac{\left|\Gamma_2\right|_{\Gamma_1}}{2}\leq \textbf{q}_{21}-\left\lfloor\frac{\textbf{q}_{21}}{\Gamma_1}\right\rfloor\Gamma_1<\left\lfloor\frac{\Gamma_1}{\left|\Gamma_2\right|_{\Gamma_1}}\right\rfloor\left|\Gamma_2\right|_{\Gamma_1}-\frac{\left|\Gamma_2\right|_{\Gamma_1}}{2},
\end{equation}
we have $1\leq n_2\leq\left\lfloor\Gamma_1/\left|\Gamma_2\right|_{\Gamma_1}\right\rfloor-1$, otherwise $n_2=0$.
\end{lemma}
\begin{proof}
Since $\textbf{q}_{21}<-|\Gamma_2|_{\Gamma_1}/2$, we have $r_1<r_2$ from Lemma \ref{lem3}, and thus $0\leq n_2\leq\left\lfloor\Gamma_1/\left|\Gamma_2\right|_{\Gamma_1}\right\rfloor-1$ from Lemma \ref{lem1}.
From (\ref{e23}), we have
\begin{equation}\label{qq}
\begin{split}
\textbf{q}_{21}-\left\lfloor\frac{\textbf{q}_{21}}{\Gamma_1}\right\rfloor\Gamma_1&=n_2|\Gamma_2|_{\Gamma_1}-k\Gamma_1+\frac{\triangle r_1-\triangle r_2}{m}-\left\lfloor\frac{\textbf{q}_{21}}{\Gamma_1}\right\rfloor\Gamma_1\quad\mbox{with }k\geq1\\
&=n_2|\Gamma_2|_{\Gamma_1}-k\Gamma_1+\frac{\triangle r_1-\triangle r_2}{m}-\left\lfloor-k+\frac{n_2|\Gamma_2|_{\Gamma_1}+(\triangle r_1-\triangle r_2)/m}{\Gamma_1}\right\rfloor\Gamma_1\\
&=n_2|\Gamma_2|_{\Gamma_1}+\frac{\triangle r_1-\triangle r_2}{m}-\left\lfloor\frac{n_2|\Gamma_2|_{\Gamma_1}+(\triangle r_1-\triangle r_2)/m}{\Gamma_1}\right\rfloor\Gamma_1.
\end{split}
\end{equation}
When $1\leq n_2\leq\left\lfloor\Gamma_1/\left|\Gamma_2\right|_{\Gamma_1}\right\rfloor-1$, it follows from (\ref{con11}) and (\ref{daxiao}) that
\begin{equation}
\left\lfloor\frac{n_2|\Gamma_2|_{\Gamma_1}+(\triangle r_1-\triangle r_2)/m}{\Gamma_1}\right\rfloor=0.
\end{equation}
Then,
$\textbf{q}_{21}-\left\lfloor \textbf{q}_{21}/\Gamma_1\right\rfloor\Gamma_1=n_2|\Gamma_2|_{\Gamma_1}+(\triangle r_1-\triangle r_2)/m$. So,
\begin{equation}
\frac{\left|\Gamma_2\right|_{\Gamma_1}}{2}\leq \textbf{q}_{21}-\left\lfloor\frac{\textbf{q}_{21}}{\Gamma_1}\right\rfloor\Gamma_1<\left\lfloor\frac{\Gamma_1}{\left|\Gamma_2\right|_{\Gamma_1}}\right\rfloor|\Gamma_2|_{\Gamma_1}-\frac{\left|\Gamma_2\right|_{\Gamma_1}}{2}.
\end{equation}
When $n_2=0$, we have
\begin{equation}
\textbf{q}_{21}-\left\lfloor\frac{\textbf{q}_{21}}{\Gamma_1}\right\rfloor\Gamma_1=\frac{\triangle r_1-\triangle r_2}{m}-\left\lfloor\frac{(\triangle r_1-\triangle r_2)/m}{\Gamma_1}\right\rfloor\Gamma_1.
\end{equation}
So, one can see that $0\leq\textbf{q}_{21}-\left\lfloor\textbf{q}_{21}/\Gamma_1\right\rfloor\Gamma_1<\left|\Gamma_2\right|_{\Gamma_1}/2$ or $\Gamma_1-\left|\Gamma_2\right|_{\Gamma_1}/2\leq \textbf{q}_{21}-\left\lfloor\textbf{q}_{21}/\Gamma_1\right\rfloor\Gamma_1<\Gamma_1$.
Therefore, the final result is derived.
\end{proof}

We next propose a simple determination algorithm for the folding integers $n_i$ from the erroneous remainders $\tilde{r}_i$ for $i=1,2$ of an integer $N$ with $0\leq N<m_1\left(1+\left\lfloor m_2/m_1\right\rfloor\left\lfloor m_1/|m_2|_{m_1}\right\rfloor\right)$ as follows.
\begin{algorithm}[H]
  \caption{\!:}
  \label{alg:Framwork2}
  \begin{algorithmic}[1]
  \State Calculate $\textbf{q}_{21}\triangleq(\tilde{r}_1-\tilde{r}_2)/m$.
   \label{code:11}
  \State If $\textbf{q}_{21}\geq|\Gamma_2|_{\Gamma_1}/2$, let
  \begin{equation}
  \hat{n}_2=\left[\frac{\textbf{q}_{21}}{|\Gamma_2|_{\Gamma_1}}\right].
  \end{equation}
  If $\textbf{q}_{21}<-|\Gamma_2|_{\Gamma_1}/2$ and
  $\frac{\left|\Gamma_2\right|_{\Gamma_1}}{2}\leq \textbf{q}_{21}-\left\lfloor\frac{\textbf{q}_{21}}{\Gamma_1}\right\rfloor\Gamma_1<\left\lfloor\frac{\Gamma_1}{\left|\Gamma_2\right|_{\Gamma_1}}\right\rfloor\left|\Gamma_2\right|_{\Gamma_1}-\frac{\left|\Gamma_2\right|_{\Gamma_1}}{2}$,
  let
  \begin{equation}
 \hat{n}_2=\left[\frac{\textbf{q}_{21}-\left\lfloor\textbf{q}_{21}/\Gamma_1\right\rfloor\Gamma_1}{|\Gamma_2|_{\Gamma_1}}\right].
  \end{equation}
  Otherwise, let $\hat{n}_2=0$.
   \label{code:21}
   \State Let
   \begin{equation}
   \hat{n}_1=\left[\frac{\hat{n}_2m_2+\tilde{r}_2-\tilde{r}_1}{m_1}\right].
   \end{equation}
  \end{algorithmic}
\end{algorithm}

Then, we have the following result.
\begin{theorem}\label{th2}
If the remainder errors satisfy
\begin{equation}\label{concon1}
-\frac{|\Gamma_2|_{\Gamma_1}}{2}\leq\frac{\triangle r_1-\triangle r_2}{m}<\frac{|\Gamma_2|_{\Gamma_1}}{2},
\end{equation}
then the dynamic range of $N$ is $m_1\left(1+\left\lfloor m_2/m_1\right\rfloor\left\lfloor m_1/|m_2|_{m_1}\right\rfloor\right)$, and the folding integers can be accurately determined in \textbf{Algorithm \ref{alg:Framwork2}}, i.e., $\hat{n}_i=n_i$ for $i=1,2$.
\end{theorem}
\begin{proof}
If $|\Gamma_2|_{\Gamma_1}=1$, it is easy to see that $m_1\left(1+\left\lfloor m_2/m_1\right\rfloor\left\lfloor m_1/|m_2|_{m_1}\right\rfloor\right)=m_1\Gamma_2=\mbox{lcm}(m_1,m_2)$. Proposition \ref{pr1} has proven that (\ref{concon1}) is a necessary and sufficient condition for accurate determination of the folding integers. In the following, we assume $|\Gamma_2|_{\Gamma_1}\geq2$ and let $N$ be an integer with $0\leq N<m_1\left(1+\left\lfloor m_2/m_1\right\rfloor\left\lfloor m_1/|m_2|_{m_1}\right\rfloor\right)$.
When $\textbf{q}_{21}\geq|\Gamma_2|_{\Gamma_1}/2$, we know $r_1>r_2$ and $(r_1-r_2)/m=n_2|\Gamma_2|_{\Gamma_1}$ from Lemma \ref{lem3}. So,
\begin{equation}
\begin{split}
\hat{n}_2&=\left[\frac{\textbf{q}_{21}}{|\Gamma_2|_{\Gamma_1}}\right]\\
&=\left[\frac{n_2|\Gamma_2|_{\Gamma_1}+(\triangle r_1-\triangle r_2)/m}{|\Gamma_2|_{\Gamma_1}}\right]\\
&=n_2+\left[\frac{(\triangle r_1-\triangle r_2)/m}{|\Gamma_2|_{\Gamma_1}}\right]\\
&=n_2.
\end{split}
\end{equation}
When $\textbf{q}_{21}<-|\Gamma_2|_{\Gamma_1}/2$ and $\frac{\left|\Gamma_2\right|_{\Gamma_1}}{2}\leq \textbf{q}_{21}-\left\lfloor\frac{\textbf{q}_{21}}{\Gamma_1}\right\rfloor\Gamma_1<\left\lfloor\frac{\Gamma_1}{\left|\Gamma_2\right|_{\Gamma_1}}\right\rfloor\left|\Gamma_2\right|_{\Gamma_1}-\frac{\left|\Gamma_2\right|_{\Gamma_1}}{2}$, we know $r_1<r_2$ and $1\leq n_2\leq\left\lfloor\Gamma_1/\left|\Gamma_2\right|_{\Gamma_1}\right\rfloor-1$ from Lemma \ref{lem2}. So,
\begin{equation}
\begin{split}
\hat{n}_2&=\left[\frac{\textbf{q}_{21}-\left\lfloor\textbf{q}_{21}/\Gamma_1\right\rfloor\Gamma_1}{|\Gamma_2|_{\Gamma_1}}\right]\\
&=\left[\frac{n_2|\Gamma_2|_{\Gamma_1}+(\triangle r_1-\triangle r_2)/m}{|\Gamma_2|_{\Gamma_1}}\right]\\
&=n_2.
\end{split}
\end{equation}
When $\textbf{q}_{21}<-|\Gamma_2|_{\Gamma_1}/2$, and $\textbf{q}_{21}-\left\lfloor\frac{\textbf{q}_{21}}{\Gamma_1}\right\rfloor\Gamma_1<\frac{\left|\Gamma_2\right|_{\Gamma_1}}{2}$ or $\textbf{q}_{21}-\left\lfloor\frac{\textbf{q}_{21}}{\Gamma_1}\right\rfloor\Gamma_1\geq\left\lfloor\frac{\Gamma_1}{\left|\Gamma_2\right|_{\Gamma_1}}\right\rfloor\left|\Gamma_2\right|_{\Gamma_1}-\frac{\left|\Gamma_2\right|_{\Gamma_1}}{2}$, we know $n_2=0$ from Lemma \ref{lem2}, and $\hat{n}_2=n_2=0$ in \textbf{Algorithm \ref{alg:Framwork2}}.
When $-|\Gamma_2|_{\Gamma_1}/2\leq \textbf{q}_{21}<|\Gamma_2|_{\Gamma_1}/2$, we have $r_1=r_2$ from Lemma \ref{lem3}. So, $n_1=n_2=\hat{n}_2=0$. Hence, we obtain $\hat{n}_2=n_2$ in \textbf{Algorithm \ref{alg:Framwork2}}. After determining $n_2$, let
\begin{equation}
\begin{split}
\hat{n}_1&=\left[\frac{\hat{n}_2m_2+\tilde{r}_2-\tilde{r}_1}{m_1}\right]\\
&=\left[\frac{N-r_1+\triangle r_2-\triangle r_1}{m_1}\right]\\
&=n_1+\left[\frac{\triangle r_2-\triangle r_1}{m_1}\right]\\
&=n_1.
\end{split}
\end{equation}
Therefore, we can accurately determine $n_i$, i.e., $\hat{n}_i=n_i$, for $i=1,2$ in the above \textbf{Algorithm \ref{alg:Framwork2}}.

We next prove that the dynamic range is indeed $m_1\left(1+\left\lfloor m_2/m_1\right\rfloor\left\lfloor m_1/|m_2|_{m_1}\right\rfloor\right)$. Suppose that the dynamic range is larger than $m_1\left(1+\left\lfloor m_2/m_1\right\rfloor\left\lfloor m_1/|m_2|_{m_1}\right\rfloor\right)$. Let
$N=m_1\left(1+\left\lfloor m_2/m_1\right\rfloor\left\lfloor m_1/|m_2|_{m_1}\right\rfloor\right)$.
Then, we have $r_1=0$ and $r_2=m_1-|m_2|_{m_1}\left\lfloor m_1/|m_2|_{m_1}\right\rfloor$ from (\ref{aaa}). Since $|\Gamma_2|_{\Gamma_1}\geq2$, we assume $\triangle r_1=\left\lfloor|\Gamma_1|_{|\Gamma_2|_{\Gamma_1}}/2\right\rfloor m$ and $\triangle r_2=0$.
It is obvious to see that $\triangle r_1$ and $\triangle r_2$ satisfy (\ref{concon1}). Following \textbf{Algorithm \ref{alg:Framwork2}}, we calculate
\begin{equation}
\begin{split}
\textbf{q}_{21}&=\frac{r_1-r_2}{m}+\frac{\triangle r_1-\triangle r_2}{m}\\
&=|\Gamma_2|_{\Gamma_1}\left\lfloor\frac{\Gamma_1}{|\Gamma_2|_{\Gamma_1}}\right\rfloor-\Gamma_1+\left\lfloor\frac{|\Gamma_1|_{|\Gamma_2|_{\Gamma_1}}}{2}\right\rfloor\\
&=-|\Gamma_1|_{|\Gamma_2|_{\Gamma_1}}+\left\lfloor\frac{|\Gamma_1|_{|\Gamma_2|_{\Gamma_1}}}{2}\right\rfloor\\
&=-\left\lceil\frac{|\Gamma_1|_{|\Gamma_2|_{\Gamma_1}}}{2}\right\rceil\geq-\frac{|\Gamma_2|_{\Gamma_1}}{2}.
\end{split}
\end{equation}
So, we have $-|\Gamma_2|_{\Gamma_1}/2\leq\textbf{q}_{21}<0$, and then we get $\hat{n}_2=0$ in \textbf{Algorithm \ref{alg:Framwork2}}. But from (\ref{aaa}), we know $n_2=\left\lfloor m_1/|m_2|_{m_1}\right\rfloor\neq0$, i.e., $\hat{n}_2\neq n_2$. Hence, we have proven that the dynamic range is $m_1\left(1+\left\lfloor m_2/m_1\right\rfloor\left\lfloor m_1/|m_2|_{m_1}\right\rfloor\right)$.
\end{proof}

Recall that $\tau$ is the remainder error bound, i.e., $|\triangle r_i|\leq\tau$ for $i=1,2$. If
\begin{equation}
\tau<\frac{m|\Gamma_2|_{\Gamma_1}}{4},
\end{equation}
we have
\begin{equation}\label{ap}
|\triangle r_1-\triangle r_2|<\frac{m|\Gamma_2|_{\Gamma_1}}{2}.
\end{equation}
Clearly, (\ref{ap}) implies the sufficiency (\ref{concon1}) in Theorem \ref{th2}. Thus, Proposition \ref{pr3} can be thought of as a corollary of Theorem \ref{th2}. More importantly, we have presented a simple closed-form algorithm, \textbf{Algorithm \ref{alg:Framwork2}}, to determine the folding integers.

\begin{example}
Let $m_1=8\cdot5$ and $m_2=8\cdot17$. When $0\leq N<8\cdot5\cdot17=680$, the robustness bound is $8/4$ from Proposition \ref{pr2}. When $0\leq N<40\cdot(1+3\cdot2)=280$, its robustness bound becomes $16/4$ from Theorem \ref{th2} or Proposition \ref{pr3}.
\end{example}

\section{Extended Robust Remaindering with Two Moduli}\label{sec4}
Similar to Proposition \ref{pr4}, we first obtain an extension of Theorem \ref{th2} if $|\Gamma_2|_{\Gamma_1}\geq2$ in this section, where the exact dynamic range with a closed-form formula is found. A closed-form determination algorithm for the folding integers is then proposed as well.

Let $\sigma_{-1}=\Gamma_2,\sigma_{0}=\Gamma_1$, and for $i\geq 1$,
\begin{equation}\label{defsig}
\sigma_i=|\sigma_{i-2}|_{\sigma_{i-1}}.
\end{equation}

\begin{lemma}\label{lem5}
For $i\geq 1$, we have
\begin{equation}
\sigma_{i-2}=\left\lfloor\frac{\sigma_{i-2}}{\sigma_{i-1}}\right\rfloor\sigma_{i-1}+\sigma_{i}.
\end{equation}
There exists an index $K$ with $K\geq0$ such that $\sigma_{K}>1$ and $\sigma_{K+1}=1$. Moreover, $\sigma_{i-1}$ and $\sigma_{i}$ are co-prime for $0\leq i\leq K+1$, and
\begin{equation}
\sigma_{-1}>\cdots>\sigma_{K}>\sigma_{K+1}=1.
\end{equation}
\end{lemma}
\begin{proof}
From the definition of $\sigma_i$ for $i\geq1$ in (\ref{defsig}), it is easy to see that
\begin{equation}\label{sss}
\sigma_{i-2}=\left\lfloor\frac{\sigma_{i-2}}{\sigma_{i-1}}\right\rfloor\sigma_{i-1}+|\sigma_{i-2}|_{\sigma_{i-1}}=\left\lfloor\frac{\sigma_{i-2}}{\sigma_{i-1}}\right\rfloor\sigma_{i-1}+\sigma_{i}.
\end{equation}
Since $\sigma_0$ and $\sigma_{-1}$ are known co-prime, and
$\sigma_{-1}=\left\lfloor\sigma_{-1}/\sigma_{0}\right\rfloor\sigma_{0}+\sigma_1$
when $i=1$ in (\ref{sss}), we obtain that $\sigma_0$ and $\sigma_1$ are co-prime.
If $\sigma_1=1$, then $K=0$ and $\sigma_{-1}>\sigma_{0}>\sigma_1=1$.
From (\ref{defsig}), we have $\sigma_1<\sigma_0$. So, if $\sigma_1>1$,
since $\sigma_0$ and $\sigma_1$ are co-prime, and $\sigma_{0}=\left\lfloor\sigma_{0}/\sigma_{1}\right\rfloor\sigma_{1}+\sigma_2$ when $i=2$ in (\ref{sss}), we obtain that $\sigma_1$ and $\sigma_2$ are co-prime. If $\sigma_2=1$, then $K=1$ and $\sigma_{-1}>\sigma_{0}>\sigma_1>\sigma_2=1$.
From (\ref{defsig}), we have $\sigma_2<\sigma_1$. So, if $\sigma_2>1$, since $\sigma_1$ and $\sigma_2$ are co-prime, and $\sigma_{1}=\left\lfloor\sigma_{1}/\sigma_{2}\right\rfloor\sigma_{2}+\sigma_3$ when $i=3$ in (\ref{sss}), we obtain that $\sigma_2$ and $\sigma_3$ are co-prime. If $\sigma_3=1$, then $K=2$ and $\sigma_{-1}>\sigma_{0}>\sigma_1>\sigma_2>\sigma_3=1$.
We continue this procedure until we find an index $K$ such that $\sigma_{K}>1$ and $\sigma_{K+1}=1$. Then, from (\ref{defsig}), we have $\sigma_K<\sigma_{K-1}$. Since $\sigma_{K-1}$ and $\sigma_K$ are co-prime, and
$\sigma_{K-1}=\left\lfloor\sigma_{K-1}/\sigma_{K}\right\rfloor\sigma_{K}+\sigma_{K+1}$ when $i=K+1$ in (\ref{sss}), we obtain that $\sigma_K$ and $\sigma_{K+1}$ are co-prime.
Moreover, $\sigma_{-1}>\sigma_{0}>\cdots>\sigma_{K}>\sigma_{K+1}=1$.
\end{proof}

\begin{lemma}\label{lem6}
$|t_1\Gamma_2|_{\Gamma_1}\neq|t_2\Gamma_2|_{\Gamma_1}$ for any pair of integers $t_1,t_2$, where $t_1\neq t_2$ and $0\leq t_1,t_2<\Gamma_1$. Also, $|t_1\Gamma_1|_{\Gamma_2}\neq|t_2\Gamma_1|_{\Gamma_2}$ for any pair of integers $t_1,t_2$, where $t_1\neq t_2$ and $0\leq t_1,t_2<\Gamma_2$.
\end{lemma}
\begin{proof}
Suppose that $|t_1\Gamma_2|_{\Gamma_1}=|t_2\Gamma_2|_{\Gamma_1}=u$ for $0\leq t_1\neq t_2<\Gamma_1$. Then, we have, for some integers $k_1,k_2$,
\begin{equation}\label{eqar1}
 t_1\Gamma_2=k_1\Gamma_1+u\quad\mbox{and}\quad t_2\Gamma_2=k_2\Gamma_1+u.
\end{equation}
From (\ref{eqar1}), we have $(t_1-t_2)\Gamma_2=(k_1-k_2)\Gamma_1$.
Since $-\Gamma_1<t_1-t_2<\Gamma_1$, and $\Gamma_1$ and $\Gamma_2$ are co-prime, we get $t_1=t_2$. This contradicts the assumption that $t_1\neq t_2$. So, $|t_1\Gamma_2|_{\Gamma_1}\neq|t_2\Gamma_2|_{\Gamma_1}$ for any pair of $t_1,t_2$, where $t_1\neq t_2$ and $0\leq t_1,t_2<\Gamma_1$. In the same way, we can prove the latter statement that $|t_1\Gamma_1|_{\Gamma_2}\neq|t_2\Gamma_1|_{\Gamma_2}$ for any pair of $t_1,t_2$, where $t_1\neq t_2$ and $0\leq t_1,t_2<\Gamma_2$.
\end{proof}

Based on Lemma \ref{lem6}, we can define a set of $S_{2,n}$ as
\begin{equation}\label{s2}
S_{2,n}\triangleq\{|t\Gamma_2|_{\Gamma_1}: t=0,1,\cdots,n,\mbox{ where }\Gamma_1>n\geq1\}
\end{equation}
and the minimum distance between any two elements in $S_{2,n}$ as $d_{2,n}$. Let
\begin{equation}\label{n2i}
\ddot{n}_{2,j}\triangleq\mbox{max}\{n: d_{2,n}\geq\sigma_j\},
\end{equation}
where $1\leq j\leq K+1$ and $K$ is defined in Lemma \ref{lem5}. Similarly, define
\begin{equation}\label{s1}
S_{1,n}\triangleq\{|t\Gamma_1|_{\Gamma_2}: t=0,1,\cdots,n,\mbox{ where }\Gamma_2>n\geq1\}
\end{equation}
and the minimum distance between any two elements in $S_{1,n}$ as $d_{1,n}$. Let
\begin{equation}\label{n1i}
\ddot{n}_{1,j}\triangleq\mbox{max}\{n: d_{1,n}\geq\sigma_j\},
\end{equation}
where $1\leq j\leq K+1$ and $K$ is defined in Lemma \ref{lem5}. Next, we obtain the values of $\ddot{n}_{2,j}$ and $\ddot{n}_{1,j}$ for $1\leq j\leq K+1$ as follows.

\begin{lemma}\label{cal}
When $K=0$, we have $\ddot{n}_{2,1}=\Gamma_1-1$. When $K\geq1$, we have
$\ddot{n}_{2,K+1}=\Gamma_1-1$ and for $1\leq j\leq K$,
\begin{equation}\label{n222}
 \ddot{n}_{2,j} =
  \begin{cases}
    \left\lfloor\frac{\Gamma_1}{\sigma_1}\right\rfloor &\quad \text{if } j=1;\\
    \left\lfloor\frac{\Gamma_1}{\sigma_1}\right\rfloor\left\lfloor\frac{\sigma_1}{\sigma_2}\right\rfloor &\quad \text{if } j=2;\\
    \left\lfloor\frac{\sigma_{2p}}{\sigma_{2p+1}}\right\rfloor(\ddot{n}_{2,2p}+1)+\ddot{n}_{2,2p-1} & \quad \text{if } j=2p+1\text{ for }p\geq1;\\
    \left\lfloor\frac{\sigma_{2p+1}}{\sigma_{2p+2}}\right\rfloor\ddot{n}_{2,2p+1}+\ddot{n}_{2,2p}  & \quad \text{if } j=2p+2\text{ for }p\geq1.\\
  \end{cases}
\end{equation}
Also, when $K=0$, we have $\ddot{n}_{1,1}=\Gamma_2-1$. When $K\geq1$, we have $\ddot{n}_{1,K+1}=\Gamma_2-1$ and for $1\leq j\leq K$,
\begin{equation}\label{n111}
 \ddot{n}_{1,j} =
  \begin{cases}
  \left\lfloor\frac{\Gamma_2}{\Gamma_1}\right\rfloor\left\lfloor\frac{\Gamma_1}{\sigma_1}\right\rfloor&\quad \text{if } j=1;\\
  \left\lfloor\frac{\Gamma_2}{\Gamma_1}\right\rfloor\left\lfloor\frac{\Gamma_1}{\sigma_1}
\right\rfloor\left\lfloor\frac{\sigma_1}{\sigma_2}\right\rfloor+\left\lfloor\frac{\sigma_1}{\sigma_2}
\right\rfloor+\left\lfloor\frac{\Gamma_2}{\Gamma_1}\right\rfloor&\quad \text{if } j=2;\\
    \left\lfloor\frac{\sigma_{2p}}{\sigma_{2p+1}}\right\rfloor\ddot{n}_{1,2p}+\ddot{n}_{1,2p-1} & \quad \text{if } j=2p+1\text{ for }p\geq1;\\
    \left\lfloor\frac{\sigma_{2p+1}}{\sigma_{2p+2}}\right\rfloor(\ddot{n}_{1,2p+1}+1)+\ddot{n}_{1,2p}  & \quad \text{if } j=2p+2\text{ for }p\geq1.\\
  \end{cases}
\end{equation}
\end{lemma}
\begin{proof}
When $K=0$, i.e., $\sigma_1=1$, we have $\ddot{n}_{2,1}=\Gamma_1-1$ from the definitions of $\ddot{n}_{2,j}$ in (\ref{n2i}) and $S_{2,n}$ in (\ref{s2}). When $K\geq1$, due to $\sigma_{K+1}=1$ we also easily get $\ddot{n}_{2,K+1}=\Gamma_1-1$. Note that for an integer $t$, $|t\Gamma_2|_{\Gamma_1}=|t|\Gamma_2|_{\Gamma_1}|_{\Gamma_1}=|t\sigma_1|_{\Gamma_1}$. Moreover, $\Gamma_1$ and $\sigma_1$ are co-prime from Lemma \ref{lem5}. So, when $0\leq t\leq\left\lfloor\Gamma_1/\sigma_1\right\rfloor$, we have
$0\leq t\sigma_1<\Gamma_1$,
and therefore, when $1\leq n\leq\left\lfloor\Gamma_1/\sigma_1\right\rfloor$, we have $S_{2,n}=\{t\sigma_1: t=0,1,\cdots,n\}$,
and $d_{2,n}=\sigma_1$. When $t=\left\lfloor\Gamma_1/\sigma_1\right\rfloor+1$, we have
\begin{equation}
\begin{split}
|t\Gamma_2|_{\Gamma_1}&=|t\sigma_1|_{\Gamma_1}\\
&=\left|\left\lfloor\frac{\Gamma_1}{\sigma_1}\right\rfloor\sigma_1+\sigma_1\right|_{\Gamma_1}\\
&=|\Gamma_1-\sigma_2+\sigma_1|_{\Gamma_1}\\
&=\sigma_1-\sigma_2.
\end{split}
\end{equation}
So, $d_{2,\left\lfloor\Gamma_1/\sigma_1\right\rfloor+1}=\mbox{min}(\sigma_2,\sigma_1-\sigma_2)<\sigma_1$, and we obtain
\begin{equation}
\ddot{n}_{2,1}=\left\lfloor\frac{\Gamma_1}{\sigma_1}\right\rfloor.
\end{equation}
When $K=1$, we have $\ddot{n}_{2,2}=\Gamma_1-1$. We next assume $K\geq2$.
One can see that the points in $S_{2,\ddot{n}_{2,1}}$ split $[0,\Gamma_1)$ into $\ddot{n}_{2,1}$ closed intervals $[i\sigma_1,i\sigma_1+\sigma_1]$ with length $\sigma_1$ for $0\leq i\leq\ddot{n}_{2,1}-1$ and one half-open interval $[\ddot{n}_{2,1}\sigma_1,\Gamma_1)$ with length $\sigma_2$, i.e., $S_{2,\ddot{n}_{2,1}}$ is composed of the beginnings and the ends of all the closed intervals with length $\sigma_1$ and the beginning of the half-open interval with length $\sigma_2$.
Each closed interval with length $\sigma_1$ will produce $\lfloor\sigma_1/\sigma_2\rfloor-1$ closed intervals with length $\sigma_2$ and one closed interval with length $\sigma_2+\sigma_3$. So, $S_{2,\ddot{n}_{2,2}}$ is composed of the beginnings and the ends of all the closed intervals with length $\sigma_2$, the beginnings and the ends of all the closed intervals with length $\sigma_2+\sigma_3$, and the beginning of the half-open interval with length $\sigma_2$. Accordingly, we have
\begin{equation}
\begin{split}
\ddot{n}_{2,2}&=\ddot{n}_{2,1}+\ddot{n}_{2,1}\left(\left\lfloor\frac{\sigma_1}{\sigma_2}\right\rfloor-1\right)\\
&=\left\lfloor\frac{\Gamma_1}{\sigma_1}\right\rfloor\left\lfloor\frac{\sigma_1}{\sigma_2}\right\rfloor.
\end{split}
\end{equation}
When $K=2$, we have $\ddot{n}_{2,3}=\Gamma_1-1$. We next assume $K\geq3$. In this stage, we have $\ddot{n}_{2,1}(\lfloor\sigma_1/\sigma_2\rfloor-1)=\ddot{n}_{2,2}-\ddot{n}_{2,1}$ closed intervals with length $\sigma_2$, $\ddot{n}_{2,1}$ closed intervals with length $\sigma_2+\sigma_3$, and one half-open interval with length $\sigma_2$. Each closed interval with length $\sigma_2+\sigma_3$ will produce one closed interval with length $\sigma_2$ and one closed interval with length $\sigma_3$.
Each closed interval with length $\sigma_2$ will produce $\lfloor\sigma_2/\sigma_3\rfloor-1$ closed intervals with length $\sigma_3$ and one closed interval with length $\sigma_3+\sigma_4$. The half-open interval with length $\sigma_2$ will produce $\lfloor\sigma_2/\sigma_3\rfloor$ closed intervals with length $\sigma_3$ and one half-open interval with length $\sigma_4$. So, $S_{2,\ddot{n}_{2,3}}$ is composed of the beginnings and the ends of all the closed intervals with length $\sigma_3$, the beginnings and the ends of all the closed intervals with length $\sigma_3+\sigma_4$, and the beginning of the half-open interval with length $\sigma_4$. Accordingly, we have
\begin{equation}
\begin{split}
\ddot{n}_{2,3}&=\ddot{n}_{2,2}+\ddot{n}_{2,1}+\left(\ddot{n}_{2,1}+\ddot{n}_{2,1}\left(\left\lfloor\frac{\sigma_1}{\sigma_2}\right\rfloor-1\right)\right)\left(\left\lfloor\frac{\sigma_2}{\sigma_3}\right\rfloor-1\right)+\left\lfloor\frac{\sigma_2}{\sigma_3}\right\rfloor\\
&=\left\lfloor\frac{\sigma_2}{\sigma_3}\right\rfloor(\ddot{n}_{2,2}+1)+\ddot{n}_{2,1}.
\end{split}
\end{equation}
When $K=3$, we have $\ddot{n}_{2,4}=\Gamma_1-1$. We next assume $K\geq4$. In this stage, we have $\ddot{n}_{2,3}-\ddot{n}_{2,2}$ closed intervals with length $\sigma_3$, $\ddot{n}_{2,2}$ closed intervals with length $\sigma_3+\sigma_4$, and one half-open interval with length $\sigma_4$. Each closed interval with length $\sigma_3+\sigma_4$ will produce one closed interval with length $\sigma_3$ and one closed interval with length $\sigma_4$. Each closed interval with length $\sigma_3$ will produce $\lfloor\sigma_3/\sigma_4\rfloor-1$ closed intervals with length $\sigma_4$ and one closed interval with length $\sigma_4+\sigma_5$. So, $S_{2,\ddot{n}_{2,4}}$ is composed of the beginnings and the ends of all the closed intervals with length $\sigma_4$, the beginnings and the ends of all the closed intervals with length $\sigma_4+\sigma_5$, and the beginning of the half-open interval with length $\sigma_4$. Accordingly, we have
\begin{equation}
\begin{split}
\ddot{n}_{2,4}&=\ddot{n}_{2,3}+\ddot{n}_{2,2}+\left(\ddot{n}_{2,2}+\ddot{n}_{2,3}-\ddot{n}_{2,2}\right)\left(\left\lfloor\frac{\sigma_3}{\sigma_4}\right\rfloor-1\right)\\
&=\left\lfloor\frac{\sigma_3}{\sigma_4}\right\rfloor\ddot{n}_{2,3}+\ddot{n}_{2,2}.
\end{split}
\end{equation}
Following the process, one can see that we can obtain the values of $\ddot{n}_{2,j}$ as in (\ref{n222}). Similarly, we can obtain the values of $\ddot{n}_{1,j}$ as in (\ref{n111}).
\end{proof}

Similar to Lemma \ref{lem3}, we have the following lemma.
\begin{lemma}\label{lem7}
Let $N$ be an integer with $0\leq N<\mbox{min}(m_2(1+\ddot{n}_{2,j}),m_1(1+\ddot{n}_{1,j}))$ for some $j$, $1\leq j\leq K+1$, and the remainder errors satisfy
\begin{equation}\label{con11hou}
-\frac{\sigma_j}{2}\leq\frac{\triangle r_1-\triangle r_2}{m}<\frac{\sigma_j}{2}.
\end{equation}
We can obtain the following three cases:
\begin{enumerate}
  \item if $\textbf{q}_{21}\geq\sigma_j/2$, we have $r_1>r_2$;
  \item if $\textbf{q}_{21}<-\sigma_j/2$, we have $r_1<r_2$;
  \item if $-\sigma_j/2\leq\textbf{q}_{21}<\sigma_j/2$, we have $r_1=r_2$.
\end{enumerate}
\end{lemma}
\begin{proof}
According to $0\leq N<\mbox{min}(m_2(1+\ddot{n}_{2,j}),m_1(1+\ddot{n}_{1,j}))$, we have
\begin{equation}\label{shuyu}
0\leq n_2\leq\ddot{n}_{2,j}\quad\mbox{and}\quad 0\leq n_1\leq\ddot{n}_{1,j}.
\end{equation}
Since $n_2\Gamma_2-n_1\Gamma_1=(r_1-r_2)/m$, we have
\begin{equation}\label{eqr1}
|n_2\Gamma_2|_{\Gamma_1}\equiv\frac{r_1-r_2}{m}\mbox{ mod }\Gamma_1.
\end{equation}
When $r_1>r_2$, then (\ref{eqr1}) becomes
\begin{equation}\label{cone}
|n_2\Gamma_2|_{\Gamma_1}=\frac{r_1-r_2}{m}.
\end{equation}
From (\ref{shuyu}), we know $|n_2\Gamma_2|_{\Gamma_1}\in S_{2,\ddot{n}_{2,j}}$. Moreover, since $0\in S_{2,\ddot{n}_{2,j}}$ and $d_{2,\ddot{n}_{2,j}}\geq\sigma_j$, one can see that
\begin{equation}\label{bbas}
\frac{r_1-r_2}{m}\geq\sigma_j.
\end{equation}
Hence, based on (\ref{bbas}) and (\ref{con11hou}), we get
\begin{equation}\label{qq21}
\frac{\tilde{r}_1-\tilde{r}_2}{m}=\frac{r_1-r_2}{m}+\frac{\triangle r_1-\triangle r_2}{m}\geq\frac{\sigma_j}{2}.
\end{equation}
When $r_1=r_2$, we have $\tilde{r}_1-\tilde{r}_2=\triangle r_1-\triangle r_2$. In this case, we have
\begin{equation}
-\frac{\sigma_j}{2}\leq\frac{\tilde{r}_1-\tilde{r}_2}{m}<\frac{\sigma_j}{2}.
\end{equation}
Since $n_1\Gamma_1-n_2\Gamma_2=(r_2-r_1)/m$, we have
\begin{equation}\label{eqr22}
|n_1\Gamma_1|_{\Gamma_2}\equiv\frac{r_2-r_1}{m}\mbox{ mod }\Gamma_2.
\end{equation}
When $r_1<r_2$, then (\ref{eqr22}) becomes
\begin{equation}\label{cone1}
|n_1\Gamma_1|_{\Gamma_2}=\frac{r_2-r_1}{m}.
\end{equation}
From (\ref{shuyu}), we know $|n_1\Gamma_1|_{\Gamma_2}\in S_{1,\ddot{n}_{1,j}}$. Moreover, since $0\in S_{1,\ddot{n}_{1,j}}$ and $d_{1,\ddot{n}_{1,j}}\geq\sigma_j$, one can see that
\begin{equation}\label{444}
\frac{r_2-r_1}{m}\geq\sigma_j.
\end{equation}
Hence, based on (\ref{444}) and (\ref{con11hou}), we get
\begin{equation}
\frac{\tilde{r}_1-\tilde{r}_2}{m}=\frac{r_1-r_2}{m}+\frac{\triangle r_1-\triangle r_2}{m}<-\frac{\sigma_j}{2}.
\end{equation}
Moreover, $\sigma_j\geq1$ for $1\leq j\leq K+1$, where $K$ is defined in Lemma \ref{lem5}. Therefore, we obtain the above three cases and complete the proof.
\end{proof}

Let $N$ be an integer with $0\leq N<\mbox{min}(m_2(1+\ddot{n}_{2,j}),m_1(1+\ddot{n}_{1,j}))$ for some $j$, $1\leq j\leq K+1$, and $\tilde{r}_i$ be its erroneous remainders for $i=1,2$.
We then have the following algorithm.
\begin{algorithm}[H]
  \caption{\!:}
  \label{alg:Framwork22}
  \begin{algorithmic}[1]
  \State Calculate $\ddot{n}_{2,j},\ddot{n}_{1,j}$ according to Lemma \ref{cal}, and then calculate the corresponding sets $S_{2,\ddot{n}_{2,j}}, S_{1,\ddot{n}_{1,j}}$ from (\ref{s2}) and (\ref{s1}).
  \State Calculate $\textbf{q}_{21}\triangleq(\tilde{r}_1-\tilde{r}_2)/m$.
  \State (\romannumeral1:) When $\textbf{q}_{21}\geq\sigma_j/2$, we find an element denoted by $s_2$ from $S_{2,\ddot{n}_{2,j}}$ as follows. If there exists an element $x$ in $S_{2,\ddot{n}_{2,j}}$ satisfying
  \begin{equation}\label{a1}
  -\frac{\sigma_j}{2}\leq \textbf{q}_{21}-x<\frac{\sigma_j}{2},
  \end{equation}
  let $s_2=x$. Otherwise, let $s_2$ be the element in $S_{2,\ddot{n}_{2,j}}$ that has the minimum distance to $\textbf{q}_{21}$. Then,
  calculate
  \begin{equation}\label{qc1}
  \hat{n}_2\equiv s_2\bar{\Gamma}_{21}\mbox{ mod }\Gamma_1
  \end{equation}
  and
  \begin{equation}
  \hat{n}_1=\left[\frac{\hat{n}_2m_2+\tilde{r}_2-\tilde{r}_1}{m_1}\right],
  \end{equation}
  where $\bar{\Gamma}_{21}$ is the modular multiplicative inverse of $\Gamma_2$ modulo $\Gamma_1$, i.e., $1\equiv\Gamma_2\bar{\Gamma}_{21}\mbox{ mod }\Gamma_1$.

  \hspace{-0.8cm}(\romannumeral2:) When $\textbf{q}_{21}<-\sigma_j/2$, we find an element denoted by $s_1$ from $S_{1,\ddot{n}_{1,j}}$ as follows. If there exists an element $y$ in $S_{1,\ddot{n}_{1,j}}$ satisfying that
  \begin{equation}\label{a2}
  -\frac{\sigma_j}{2}\leq \textbf{q}_{21}+y<\frac{\sigma_j}{2},
  \end{equation}
  let $s_1=y$. Otherwise, let $s_1$ be the element in $S_{1,\ddot{n}_{1,j}}$ that has the minimum distance to $-\textbf{q}_{21}$. Then,
  calculate
  \begin{equation}\label{qc2}
   \hat{n}_1\equiv s_1\bar{\Gamma}_{12}\mbox{ mod }\Gamma_2
   \end{equation}
   and
    \begin{equation}
  \hat{n}_2=\left[\frac{\hat{n}_1m_1+\tilde{r}_1-\tilde{r}_2}{m_2}\right],
  \end{equation}
   where $\bar{\Gamma}_{12}$ is the modular multiplicative inverse of $\Gamma_1$ modulo $\Gamma_2$, i.e., $1\equiv\Gamma_1\bar{\Gamma}_{12}\mbox{ mod }\Gamma_2$.

\hspace{-0.8cm}(\romannumeral3:) When $-\sigma_j/2\leq \textbf{q}_{21}<\sigma_j/2$, we let $\hat{n}_1=\hat{n}_2=0$.
  \end{algorithmic}
\end{algorithm}

Then, we have the following result.
\begin{theorem}\label{th1}
For some $j$, $1\leq j\leq K+1$, if the remainder errors satisfy
\begin{equation}\label{concon}
-\frac{\sigma_j}{2}\leq\frac{\triangle r_1-\triangle r_2}{m}<\frac{\sigma_j}{2},
\end{equation}
then the dynamic range of $N$ is $\mbox{min}(m_2(1+\ddot{n}_{2,j}),m_1(1+\ddot{n}_{1,j}))$, and the folding integers can be accurately determined in \textbf{Algorithm \ref{alg:Framwork22}}, i.e., $\hat{n}_i=n_i$ for $i=1,2$.
\end{theorem}
\begin{proof}
Let $N$ be an integer with $0\leq N<\mbox{min}(m_2(1+\ddot{n}_{2,j}),m_1(1+\ddot{n}_{1,j}))$. Then, we have
\begin{equation}
0\leq n_2\leq\ddot{n}_{2,j}<\Gamma_1\quad\mbox{and}\quad 0\leq n_1\leq\ddot{n}_{1,j}<\Gamma_2.
\end{equation}
By B\'{e}zout's lemma in
$n_i\Gamma_i-n_l\Gamma_l=(r_l-r_i)/m\mbox{ for }1\leq i\neq l\leq2$,
the folding integers $n_i$ for $i=1,2$ can be determined by
\begin{equation}\label{qiun1}
n_1\equiv\frac{r_2-r_1}{m}\bar{\Gamma}_{12}\mbox{ mod }\Gamma_2\quad\mbox{ and }\quad n_2\equiv\frac{r_1-r_2}{m}\bar{\Gamma}_{21}\mbox{ mod }\Gamma_1,
\end{equation}
where $\bar{\Gamma}_{ij}$ is the modular multiplicative inverse of $\Gamma_i$ modulo $\Gamma_j$, i.e., $1\equiv\Gamma_i\bar{\Gamma}_{ij}\mbox{ mod }\Gamma_j$.
When $\textbf{q}_{21}\geq\sigma_j/2$, we know $r_1>r_2$ based on Lemma \ref{lem7}. From (\ref{qq21}) and (\ref{concon}), we have
\begin{equation}\label{ss1}
-\frac{\sigma_j}{2}\leq \textbf{q}_{21}-\frac{r_1-r_2}{m}=\frac{\triangle r_1-\triangle r_2}{m}<\frac{\sigma_j}{2}.
\end{equation}
One can see from (\ref{cone}) that $(r_1-r_2)/m=|n_2\Gamma_2|_{\Gamma_1}\in S_{2,\ddot{n}_{2,j}}$. Next, we prove that $(r_1-r_2)/m$ is a unique element in $S_{2,\ddot{n}_{2,j}}$ to satisfy (\ref{ss1}). For any element $s\in S_{2,\ddot{n}_{2,j}}$ with $s\neq(r_1-r_2)/m$,
\begin{equation}
\textbf{q}_{21}-s=\textbf{q}_{21}-\frac{r_1-r_2}{m}+\frac{r_1-r_2}{m}-s.
\end{equation}
Since $|(r_1-r_2)/m-s|\geq\sigma_j$, we have $\textbf{q}_{21}-s\geq\sigma_j/2$ or $\textbf{q}_{21}-s<-\sigma_j/2$. Therefore, we can find a unique element $s_2$ in $S_{2,\ddot{n}_{2,j}}$ satisfying (\ref{a1}) in \textbf{Algorithm \ref{alg:Framwork22}}, and $s_2=|n_2\Gamma_2|_{\Gamma_1}=(r_1-r_2)/m$.
From (\ref{qiun1}), we have $\hat{n}_2=n_2$ in (\ref{qc1}), and
\begin{equation}
\left[\frac{\hat{n}_2m_2+\tilde{r}_2-\tilde{r}_1}{m_1}\right]=n_1+\left[\frac{\triangle r_2-\triangle r_1}{m}\right]=n_1.
\end{equation}
Similarly, when $\textbf{q}_{21}<-\sigma_j/2$, we know $r_1<r_2$ based on Lemma \ref{lem7}. From (\ref{qq21}) and (\ref{concon}), we have
\begin{equation}\label{ss2}
-\frac{\sigma_j}{2}\leq \textbf{q}_{21}+\frac{r_2-r_1}{m}=\frac{\triangle r_1-\triangle r_2}{m}<\frac{\sigma_j}{2}.
\end{equation}
One can see from (\ref{cone1}) that $(r_2-r_1)/m=|n_1\Gamma_1|_{\Gamma_2}\in S_{1,\ddot{n}_{1,j}}$. Next, we prove that $(r_2-r_1)/m$ is a unique element in $S_{1,\ddot{n}_{1,j}}$ to satisfy (\ref{ss2}). For any element $s\in S_{1,\ddot{n}_{1,j}}$ with $s\neq(r_2-r_1)/m$,
\begin{equation}
\textbf{q}_{21}+s=\textbf{q}_{21}+\frac{r_2-r_1}{m}+s-\frac{r_2-r_1}{m}.
\end{equation}
Since $|s-(r_2-r_1)/m|\geq\sigma_j$, we have $\textbf{q}_{21}+s\geq\sigma_j/2$ or $\textbf{q}_{21}+s<-\sigma_j/2$. Therefore, we can find a unique element $s_1$ in $S_{1,\ddot{n}_{1,j}}$ satisfying (\ref{a2}) in \textbf{Algorithm \ref{alg:Framwork22}}, and $s_1=|n_1\Gamma_1|_{\Gamma_2}=(r_2-r_1)/m$. From (\ref{qiun1}), we have $\hat{n}_1=n_1$ in (\ref{qc2}), and
\begin{equation}
\left[\frac{\hat{n}_1m_1+\tilde{r}_1-\tilde{r}_2}{m_2}\right]=n_2+\left[\frac{\triangle r_1-\triangle r_2}{m}\right]=n_2.
\end{equation}
Finally, when $-\sigma_j/2\leq \textbf{q}_{21}<\sigma_j/2$, we have $r_1=r_2$ based on Lemma \ref{lem7}. Then, we know $n_1=n_2=0$. So, $\hat{n}_1=\hat{n}_2=n_1=n_2=0$. Therefore, we can accurately determine $n_i$, i.e., $\hat{n}_i=n_i$, for $i=1,2$ in the above \textbf{Algorithm \ref{alg:Framwork22}}.

Next, we prove that the dynamic range is indeed $\mbox{min}(m_2(1+\ddot{n}_{2,j}),m_1(1+\ddot{n}_{1,j}))$. Without loss of generality, we assume $m_2(1+\ddot{n}_{2,j})<m_1(1+\ddot{n}_{1,j})$. Suppose that the dynamic range is larger than $m_2(1+\ddot{n}_{2,j})$. Let $N=m_2(1+\ddot{n}_{2,j})$, and we have $r_2=0$. Then, from the definition of $\ddot{n}_{2,j}$ in (\ref{n2i}), there exists an element $w$ in $S_{2,\ddot{n}_{2,j}}$ such that $|mw-r_1|<m\sigma_j$. Let $\triangle r_2=0$ and $\triangle r_1=(mw-r_1)/2$.
One can see that $\triangle r_1$ and $\triangle r_2$ satisfy (\ref{concon}). Due to $w\geq\sigma_j$, we have $\textbf{q}_{21}=(r_1+\triangle r_1)/m\geq\sigma_j/2$, and then $-\sigma_j/2<\textbf{q}_{21}-w=-\triangle r_1/m<\sigma_j/2$. For any other element $s\in S_{2,\ddot{n}_{2,j}}$, we get $\textbf{q}_{21}-s>\sigma_j/2$ or $\textbf{q}_{21}-s<-\sigma_j/2$ according to $|w-s|\geq\sigma_j$. So, $w$ is a unique element in $S_{2,\ddot{n}_{2,j}}$ to satisfy (\ref{a1}) in \textbf{Algorithm \ref{alg:Framwork22}}. However, the obtained element $w\in S_{2,\ddot{n}_{2,j}}$ does not equal $(r_1-r_2)/m$, since $(r_1-r_2)/m=|(1+\ddot{n}_{2,j})\Gamma_2|_{\Gamma_1}$ does not belong to $S_{2,\ddot{n}_{2,j}}$. Hence, $\hat{n}_2\neq n_2$ in (\ref{qc1}), and we have proven that the dynamic range is $\mbox{min}(m_2(1+\ddot{n}_{2,j}),m_1(1+\ddot{n}_{1,j}))$.
\end{proof}

\begin{corollary}\label{cor1}
For some $j$, $1\leq j\leq K+1$, if the remainder error bound $\tau$ satisfies
\begin{equation}
\tau<\frac{m\sigma_j}{4},
\end{equation}
then the dynamic range of $N$ is $\mbox{min}(m_2(1+\ddot{n}_{2,j}),m_1(1+\ddot{n}_{1,j}))$, and the folding integers can be accurately determined in \textbf{Algorithm \ref{alg:Framwork22}}, i.e., $\hat{n}_i=n_i$ for $i=1,2$.
\end{corollary}
\begin{proof}
Since $|\triangle r_i|\leq\tau<m\sigma_j/4$ for $i=1,2$, we have
\begin{equation}
|\triangle r_1-\triangle r_2|<\frac{m\sigma_j}{2}.
\end{equation}
So, from Theorem \ref{th1}, the corollary is proved.
\end{proof}

\begin{remark}\label{rem1}
Since $\ddot{n}_{2,K+1}=\Gamma_1-1$ and $\ddot{n}_{1,K+1}=\Gamma_2-1$ as in Lemma \ref{cal}, we have $\mbox{min}(m_2(1+\ddot{n}_{2,K+1}),m_1(1+\ddot{n}_{1,K+1}))=\mbox{lcm}(m_1,m_2)$. So,
when $j=K+1$, i.e., $\sigma_{j}=1$, Theorem \ref{th1} coincides with Proposition \ref{pr1}, and Corollary \ref{cor1} coincides with Proposition \ref{pr2}. In other words, when the dynamic range increases to the maximum, i.e., the lcm of the two moduli, the robustness bound decreases to the quarter of the gcd of the two moduli.
\end{remark}

Next, we prove that when $j=1$, Theorem \ref{th1} coincides with Theorem \ref{th2}.

\begin{corollary}\label{cor2}
Theorem \ref{th2} is a special case of Theorem \ref{th1} when $j=1$.
\end{corollary}
\begin{proof}
If $\sigma_1=1$, i.e., $K=0$, as described in Remark \ref{rem1}, Theorem \ref{th2} is a special case of Theorem \ref{th1} when $j=1$. In the following, we assume $\sigma_1>1$, and we only need to prove $\mbox{min}(m_2(1+\ddot{n}_{2,1}),m_1(1+\ddot{n}_{1,1}))=m_1(1+\lfloor\Gamma_2/\Gamma_1\rfloor\lfloor\Gamma_1/\sigma_1\rfloor)$. Since $\ddot{n}_{2,1}=\left\lfloor\Gamma_1/\sigma_1\right\rfloor$ and $\ddot{n}_{1,1}=\left\lfloor\Gamma_2/\Gamma_1\right\rfloor\left\lfloor\Gamma_1/\sigma_1\right\rfloor$ in Lemma \ref{cal}, we next prove $m_1(1+\ddot{n}_{1,1})<m_2(1+\ddot{n}_{2,1})$.
It is readily seen that
\begin{equation}
\begin{split}
m_1\left(1+\left\lfloor\frac{\Gamma_2}{\Gamma_1}\right\rfloor\left\lfloor\frac{\Gamma_1}{\sigma_1}\right\rfloor\right)&=m_2\frac{\Gamma_1}{\Gamma_2}\left(1+\left\lfloor\frac{\Gamma_2}{\Gamma_1}\right\rfloor\left\lfloor\frac{\Gamma_1}{\sigma_1}\right\rfloor\right)\\
&=m_2\left(\frac{\Gamma_1}{\Gamma_2}+\frac{\Gamma_1}{\Gamma_2}\left\lfloor\frac{\Gamma_2}{\Gamma_1}\right\rfloor\left\lfloor\frac{\Gamma_1}{\sigma_1}\right\rfloor\right)\\
&<m_2\left(1+\left\lfloor\frac{\Gamma_1}{\sigma_1}\right\rfloor\right).
\end{split}
\end{equation}
Therefore, $\mbox{min}(m_2(1+\ddot{n}_{2,1}),m_1(1+\ddot{n}_{1,1}))=m_1(1+\lfloor\Gamma_2/\Gamma_1\rfloor\lfloor\Gamma_1/\sigma_1\rfloor)$, and we complete the proof.
\end{proof}

\begin{example}\label{ex2}
Let $m_1=13\cdot18$ and $m_2=13\cdot29$. The lcm of the moduli is $\mbox{lcm}(m_1,m_2)=6786$. According to Lemma \ref{cal} and Corollary \ref{cor1}, we have the following result in Table \ref{table1}, where the last row, i.e., Level \uppercase\expandafter{\romannumeral1}, is the known result in Proposition \ref{pr2}.
\begin{table}[H]
\centering  
\begin{tabular}{lllccl}  
\hline
level& value of $\sigma_j$ &robustness bound&$\ddot{n}_{1,j}$&$\ddot{n}_{2,j}$ &dynamic range\\ \hline  
\uppercase\expandafter{\romannumeral5}&$\sigma_1=11$ &$\tau<(13\cdot11)/4=35.75$ &1&1&$0\leq N<468$ \\         
\uppercase\expandafter{\romannumeral4}&$\sigma_2=7$ &$\tau<(13\cdot7)/4=22.75$ &3&1&$0\leq N<754$\\        
\uppercase\expandafter{\romannumeral3}&$\sigma_3=4$ &$\tau<(13\cdot4)/4=13$ &4&3&$0\leq N<1170$ \\
\uppercase\expandafter{\romannumeral2}&$\sigma_4=3$ &$\tau<(13\cdot3)/4=9.75$ &8&4&$0\leq N<1885$ \\
\uppercase\expandafter{\romannumeral1}&$\sigma_5=1$ &$\tau<(13\cdot1)/4=3.25$ &28&17&$0\leq N<6786$ \\
\hline
\end{tabular}
\caption{The relationship between the dynamic range and the robustness bound.}
\label{table1}
\end{table}
\end{example}

Let us recall the intuitive explantation of robust reconstruction by using the method of integer position representation on the two dimensional remainder plane introduced in \textit{Part} $B$ of Section \ref{sec2}. Via the CRT, we know that the integers within $[0,\mbox{lcm}(m_1,m_2))$ and their remainders $(r_1,r_2)$ are isomorphic, i.e., different integers within $[0,\mbox{lcm}(m_1,m_2))$ have different position representations on the remainder plane.
Since $N=n_im_i+r_i$ for $i=1,2$, we have
\begin{equation}\label{line}
r_2=r_1+(n_1m_1-n_2m_2).
\end{equation}
So, all the integers $N$ (or equivalently $(r_1, r_2)$) from $0$ to $\mbox{lcm}(m_1,m_2)-1$ are connected by the lines (\ref{line}) with the slope of $1$, as depicted in Fig. \ref{figone}. Moreover, due to $0\leq N<\mbox{lcm}(m_1,m_2)$, we know
$0\leq n_1<\Gamma_2$ and $0\leq n_2<\Gamma_1$.
Then, the folding integers $n_i$ for $i=1,2$ are determined by the value of $r_2-r_1$ as in (\ref{qiun1}).
Therefore, the integers on a slanted line (\ref{line}) have the same folding integers, and every such slanted line corresponds to a unique pair of folding integers. The idea of finding the closest slanted line to the erroneous remainders $(\tilde{r}_1,\tilde{r}_2)$ in \cite{new} is equivalent to that of determining the folding integers in \cite{wjwang2010} and also this paper, i.e.,
the closest slanted line to $(\tilde{r}_1,\tilde{r}_2)$ is the line that contains the true remainders $(r_1,r_2)$, which means that the folding integers are accurately determined.
As shown in Fig. \ref{figtwo}, the robustness bound depends on the minimum distance between the set of slanted lines, and the distance between the slanted lines can be determined by their horizontal or vertical distance. So, all the integers within $[0,\mbox{min}(m_2(1+\ddot{n}_{2,j}),m_1(1+\ddot{n}_{1,j})))$ in Corollary \ref{cor1} are connected by $\ddot{n}_{2,j}+\ddot{n}_{1,j}-1$ slanted lines on the remainder plane, which include the identity line (i.e., $r_1=r_2$) denoted by $S$, $\ddot{n}_{2,j}-1$ slanted lines (i.e., $r_1>r_2$) denoted by $S_2$ below the identity line, and $\ddot{n}_{1,j}-1$ slanted lines (i.e., $r_1<r_2$) denoted by $S_1$ above the identity line. One can see that
$md_{2,\ddot{n}_{2,j}}$ is the minimum (horizontal) distance between the set $S_2\bigcup S$ of slanted lines, and $md_{1,\ddot{n}_{1,j}}$ is the minimum (vertical) distance between the set $S_1\bigcup S$ of slanted lines, where $d_{2,\ddot{n}_{2,j}}\geq\sigma_j$ and $d_{1,\ddot{n}_{1,j}}\geq\sigma_j$ are obtained in (\ref{s2}) and (\ref{s1}), respectively. Since the two sets $S_1,S_2$ are separated by the identity line in $S$ on the remainder plane, the minimum distance between all of the slanted lines is greater than or equal to $\sigma_j$. This gives an intuitive explantation of Corollary \ref{cor1}.

\section{Multi-Modular Systems and Generalization}\label{sec5}
In this section, the above newly obtained two-modular results are first applied to robust reconstruction for multi-modular systems by using cascade or parallel architectures, and then generalized from integers to real numbers.

\subsection{Robust Reconstruction for Multi-Modular Systems}
Let $m_1,m_2,\cdots,m_L$ be $L$ moduli and split into two groups:
$\{m_{1,1},\cdots,m_{1,L_1}\}\mbox{ and }\{m_{2,1},\cdots,m_{2,L_2}\}$, where $L>2$ and the two groups do not have to be disjoint, i.e., $L_1+L_2\geq L$. Let $N$ be an integer with $0\leq N<\mbox{lcm}(m_1,m_2,\cdots,m_L)$, and we can uniquely reconstruct $N$ in the following cascade process.
For $i=1,2$
and Group $i$, we first write
\begin{equation}\label{eachgroup}
\left\{\begin{array}{ll}
N_i=h_{i,k}m_{i,k}+r_{i,k}\\
0\leq N_i<\eta_i\triangleq\mbox{lcm}(m_{i,1},m_{i,2},\cdots,m_{i,L_i})\\
1\leq k \leq L_i,
\end{array}\right.
\end{equation}
and then regard $N_i$ as the remainders of the following system of congruences:
\begin{equation}\label{twostage}
\left\{\begin{array}{ll}
N=l_1\eta_1+N_1\\
N=l_2\eta_2+N_2\\
0\leq N<\mbox{lcm}(\eta_1,\eta_2)=\mbox{lcm}(m_1,m_2,\cdots,m_L).
\end{array}\right.
\end{equation}
Without loss of generality, we assume $\eta_1<\eta_2$. Replacing $N_1$ and $N_2$ in (\ref{twostage}) by (\ref{eachgroup}), we have, for $1\leq k\leq L_i$ and $i=1,2$,
\begin{equation}\label{one}
N=\left(l_i\frac{\eta_i}{m_{i,k}}+h_{i,k}\right)m_{i,k}+r_{i,k}.
\end{equation}
Assume that the remainders $r_{i,k}$ have errors:
\begin{equation}
0\leq\tilde{r}_{i,k}<m_{i,k}\quad\mbox{and}\quad|\tilde{r}_{i,k}-r_{i,k}|\leq\tau_i,
\end{equation}
where $\triangle r_{i,k}\triangleq\tilde{r}_{i,k}-r_{i,k}$ denotes the remainder error, and
$\tau_i$ denotes the remainder error bound for the remainders in the $i$-th group for $i=1,2$.
One can see from (\ref{one}) that if we can accurately determine $h_{i,k}$ and $l_i$, we can accurately determine the folding integers of $N$ modulo $m_{i,k}$.
Therefore, we first apply the robust CRT (Proposition \ref{pr2} in \cite{wjwang2010} or multi-stage robust CRT in \cite{xiaoxia1}) to each group in (\ref{eachgroup}), and obtain accurate $h_{i,k}$ and robust reconstructions $\hat{N}_i$ for $1\leq k\leq L_i$ and $i=1,2$. With these robust reconstructions from the two groups, the above newly obtained two-modular results are then applied across the two groups in (\ref{twostage}).

In what follows, let us consider without loss of generality a special case when the remaining integers of the moduli in each group factorized by their gcd are pairwise co-prime, i.e., for $i=1,2$ and Group $i$, moduli $m_{i,k}=m^{(i)}\Gamma_{i,k}$ for $1\leq k\leq L_i$, where $\Gamma_{i,1},\Gamma_{i,2},\cdots,\Gamma_{i,L_i}$ are pairwise co-prime. Denote by $m$ the gcd of $\eta_1$ and $\eta_2$, where $\eta_i$ is the lcm of all the moduli in Group $i$ and $\eta_i=m^{(i)}\prod_{k=1}^{L_i}\Gamma_{i,k}$ for $i=1,2$. We write $\eta_1=m\Gamma_1$ and $\eta_2=m\Gamma_2$, where $\Gamma_1$ and $\Gamma_2$ are co-prime and $\Gamma_1<\Gamma_2$. Then, $\ddot{n}_{2,j}$ and $\ddot{n}_{1,j}$ can be calculated according to Lemma \ref{cal}, and we have the following result.

\begin{theorem}\label{th22}
Let $N$ be an integer with $0\leq N<\mbox{min}(\eta_2(1+\ddot{n}_{2,j}),\eta_1(1+\ddot{n}_{1,j}))$ for some $j$, $1\leq j\leq K+1$. If the remainder error bounds $\tau_1$ and $\tau_2$ satisfy
\begin{equation}\label{92con}
\tau_1<\frac{m^{(1)}}{4},\;\; \tau_2<\frac{m^{(2)}}{4},\;\; \mbox{and }\tau_1+\tau_2<\frac{m\sigma_j}{2},
\end{equation}
we can accurately determine the folding integers of $N$ modulo $m_{i,k}$ for $1\leq k\leq L_i$ and $i=1,2$.
\end{theorem}
\begin{proof}
For $i=1,2$, according to Proposition \ref{pr2}, when $\tau_i<m^{(i)}/4$, we can accurately determine $h_{i,k}$ for $1\leq k\leq L_i$ in the system of congruence equations (\ref{eachgroup}), and thereby obtain robust reconstructions $\hat{N}_i$, i.e., $|\triangle N_i|=|\hat{N}_i-N_i|\leq\tau_i<m^{(i)}/4$. Then, in the system of congruence equations (\ref{twostage}), $\hat{N}_1$ and $\hat{N}_2$ become the erroneous remainders of $N$ modulo $\eta_1$ and $\eta_2$, respectively. So, from Theorem \ref{th1}, we can accurately determine $l_i$ for $i=1,2$, when $\tau_1+\tau_2<m\sigma_j/2$. Thus, by (\ref{one}), if (\ref{92con}) holds, the folding integers of $N$ modulo $m_{i,k}$ for $1\leq k\leq L_i$ and $i=1,2$ can be accurately determined.
\end{proof}

Recall that $\tau$ is the remainder error bound for all the remainders $r_{i,k}$, i.e., $|\triangle r_{i,k}|\leq\tau$ for $1\leq k\leq L_i$ and $i=1,2$. According to Theorem \ref{th22}, the following corollary is immediate.
\begin{corollary}\label{c222222}
Let $N$ be an integer with $0\leq N<\mbox{min}(\eta_2(1+\ddot{n}_{2,j}),\eta_1(1+\ddot{n}_{1,j}))$ for some $j$, $1\leq j\leq K+1$. If the remainder error bound $\tau$ satisfies
\begin{equation}\label{xiac}
\tau<\frac{\mbox{min}(m^{(1)},m^{(2)},m\sigma_j)}{4},
\end{equation}
we can accurately determine the folding integers of $N$ modulo $m_{i,k}$ for $1\leq k\leq L_i$ and $i=1,2$.
\end{corollary}
\begin{proof}
From (\ref{xiac}), we have $|\triangle r_{1,k}|\leq\tau< m^{(1)}/4$ for $1\leq k\leq L_1$, $|\triangle r_{2,k}|\leq\tau< m^{(2)}/4$ for $1\leq k\leq L_2$, and $2\tau<m\sigma_j/2$. From Theorem \ref{th22}, the corollary is proven.
\end{proof}
\begin{remark}
For a general case, i.e., the remaining integers of the moduli in each group factorized by their gcd are not necessarily pairwise co-prime, we need to use the multi-stage robust CRT introduced in \cite{xiaoxia1} to obtain accurate $h_{i,k}$ and robust reconstructions $\hat{N}_i$ for $1\leq k\leq L_i$ and $i=1,2$ in the system of congruence equations (\ref{eachgroup}).
\end{remark}

\begin{example}\label{ex3}
Let $m_1=60\cdot2, m_2=60\cdot5, m_3=70\cdot3, m_4=70\cdot7$. From the single stage robust CRT in \cite{xiaoxia1}, the dynamic range is the lcm of all the moduli, i.e., $\mbox{lcm}(m_1,m_2,m_3,m_4)=29400$, and the robustness bound is $10/4$. We split the moduli into two groups: $\{m_1,m_2\}$ and $\{m_3,m_4\}$, and use the cascade architecture to reconstruct $N$. One can see that $m^{(1)}=60$, $m^{(2)}=70$, $\eta_1=\mbox{lcm}(m_1,m_2)=600=30\cdot20$, $\eta_2=\mbox{lcm}(m_3,m_4)=1470=30\cdot49$, and $m=30$.
Then, from the two-stage robust CRT in \cite{xiaoxia1}, the dynamic range is still the lcm of all the moduli, i.e., $29400$, while the robustness bound becomes $30/4$. Let $j=2$ in Corollary \ref{c222222}, we have $\sigma_2=2$ and $\ddot{n}_{2,2}=8, \ddot{n}_{1,2}=22$ from Lemma \ref{cal}, and the robustness bound can reach $60/4$ when $0\leq N<\mbox{min}(\eta_2(1+\ddot{n}_{2,2}),\eta_1(1+\ddot{n}_{1,2}))=13230$.
\end{example}

\subsection{Generalization to Reals}
The above studies are all for integers. As we know, the robust CRT in Propositions \ref{pr1} and \ref{pr2} were naturally generalized to real numbers in \cite{wjwang2010}. The results in Propositions \ref{pr3} and \ref{pr4} were also applicable to real numbers in \cite{new}. In this section, therefore, we generalize the above Theorem \ref{th1} and Corollary \ref{cor1} to real numbers. In the following, we use boldface symbols to denote the corresponding real variables of non-boldface integer variables.

Let $\textbf{N}$ be a nonnegative real number, and $\textbf{m}_i=\textbf{m}\Gamma_i$ for $i=1,2$ be the real-valued moduli, where $\Gamma_1<\Gamma_2$ are co-prime integers. Then, $\textbf{N}$ can be expressed as
\begin{equation}\label{bbbb}
\textbf{N}=n_i\textbf{m}\Gamma_i+\textbf{r}_i,\quad i=1,2,
\end{equation}
where $n_i$ is an unknown integer (or folding integer) and $\textbf{r}_i$ is the real-valued remainder with $0\leq\textbf{r}_i<\textbf{m}_i$ for $i=1,2$.
Similarly, we assume that the remainders have errors:
\begin{equation}
0\leq\tilde{\textbf{r}}_i<\textbf{m}_i\quad\mbox{and}\quad\triangle\textbf{r}_i\triangleq\tilde{\textbf{r}}_i-\textbf{r}_i.
\end{equation}
We then have the following result.
\begin{theorem}\label{th3}
For some $j$, $1\leq j\leq K+1$, if the remainder errors satisfy
\begin{equation}\label{conconcon}
-\frac{\sigma_j}{2}\leq\frac{\triangle \textbf{r}_1-\triangle \textbf{r}_2}{\textbf{m}}<\frac{\sigma_j}{2},
\end{equation}
then the dynamic range of $\textbf{N}$ is $\mbox{min}(\textbf{m}_2(1+\ddot{n}_{2,j}),\textbf{m}_1(1+\ddot{n}_{1,j}))$, and the folding integers can be accurately determined in \textbf{Algorithm \ref{alg:Framwork22}}, i.e., $\hat{n}_i=n_i$ for $i=1,2$. In particular, if the remainder error bound $\bm{\tau}$ satisfies
\begin{equation}
|\triangle \textbf{r}_i|\leq\bm{\tau}<\frac{\textbf{m}\sigma_j}{4},
\end{equation}
then the dynamic range of $\textbf{N}$ is $\mbox{min}(\textbf{m}_2(1+\ddot{n}_{2,j}),\textbf{m}_1(1+\ddot{n}_{1,j}))$, and the folding integers can be accurately determined in \textbf{Algorithm \ref{alg:Framwork22}}, i.e., $\hat{n}_i=n_i$ for $i=1,2$.
\end{theorem}
\begin{proof}
From $0\leq\textbf{N}<\mbox{min}(\textbf{m}_2(1+\ddot{n}_{2,j}),\textbf{m}_1(1+\ddot{n}_{1,j}))$, we have
\begin{equation}
0\leq n_2\leq\ddot{n}_{2,j}<\Gamma_1\quad\mbox{ and }\quad0\leq n_1\leq\ddot{n}_{1,j}<\Gamma_2.
\end{equation}
By B\'{e}zout's lemma in
\begin{equation}
n_i\Gamma_i-n_l\Gamma_l=\frac{\textbf{r}_l-\textbf{r}_i}{\textbf{m}}\quad\mbox{for }1\leq i\neq l\leq2,
\end{equation}
the folding integers $n_i$ for $i=1,2$ can be determined by
\begin{equation}\label{qiun11}
n_1\equiv\frac{\textbf{r}_2-\textbf{r}_1}{\textbf{m}}\bar{\Gamma}_{12}\mbox{ mod }\Gamma_2\quad\mbox{ and }\quad n_2\equiv\frac{\textbf{r}_1-\textbf{r}_2}{\textbf{m}}\bar{\Gamma}_{21}\mbox{ mod }\Gamma_1,
\end{equation}
where $\bar{\Gamma}_{ij}$ is the modular multiplicative inverse of $\Gamma_i$ modulo $\Gamma_j$, i.e., $1\equiv\Gamma_i\bar{\Gamma}_{ij}\mbox{ mod }\Gamma_j$. Since
\begin{equation}
\textbf{q}_{21}\triangleq\frac{\tilde{\textbf{r}}_1-\tilde{\textbf{r}}_2}{\textbf{m}}=\frac{\textbf{r}_1-\textbf{r}_2}{\textbf{m}}+\frac{\triangle \textbf{r}_1-\triangle \textbf{r}_2}{\textbf{m}},
\end{equation}
it is easy to see from the proof of Lemma \ref{lem7} that the result in Lemma \ref{lem7} also holds for real numbers. Then, following the proofs of Theorem \ref{th1} and Corollary \ref{cor1}, we obtain the final result as desired.
\end{proof}
\begin{remark}
When $j=1$ in Theorem \ref{th3}, a generalization of Theorem \ref{th2} to real numbers is obtained. In this case, the simple \textbf{Algorithm \ref{alg:Framwork2}} is also applicable to real numbers.
\end{remark}

\section{Simulations}\label{sec6}
In this section, we present some simulation results to demonstrate the performance of our proposed algorithms. Let us first consider the case in Example \ref{ex2}, i.e., $m_1=13\cdot18$ and $m_2=13\cdot29$. In this simulation, we consider the remainder error bound $\tau$ from $0$ to $50$, and the remainder errors are uniformly distributed on $[-\tau,\tau]$.
For each level in Table \ref{table1},
the unknown integer $N$ is chosen uniformly at random and less than the corresponding dynamic range, and $2000000$ trials are implemented for each of $\tau$. We use \textbf{Algorithm \ref{alg:Framwork2}} to get the estimate $\hat{N}$ for Level $\uppercase\expandafter{\romannumeral5}$, and \textbf{Algorithm \ref{alg:Framwork22}} for other levels. Fig. \ref{fig3} shows
the mean absolute error $E(|\hat{N}-N|)$ between the estimate $\hat{N}$ and the true $N$ for each level in Table \ref{table1}. One can see that the simulation result matches well with the theoretical analysis in Theorem \ref{th1} or Corollary \ref{cor1}. When the remainder error bound is less than or equal to the robustness bound, the curve of the mean absolute error $E(|\hat{N}-N|)$ is below the curve of the remainder error bound. When the remainder error bound is larger than the robustness bound, the mean absolute error $E(|\hat{N}-N|)$ starts to deviate from the previous line trend and finally breaks the linear error bound, i.e., robust reconstruction may not hold. In Fig. \ref{fig5}, we show the mean relative error $E(|\hat{N}-N|/N)$ between the estimate $\hat{N}$ and the true $N$ for each level in Table \ref{table1}.

\begin{figure}[H]
\centerline{\includegraphics[width=1\columnwidth,draft=false]{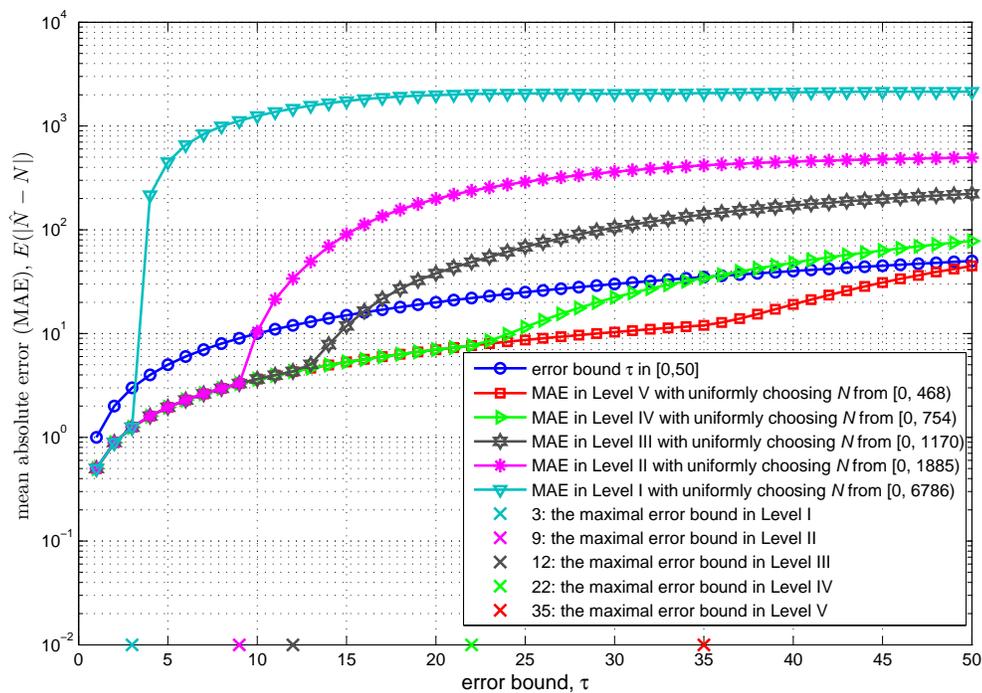}}
\vspace{-0.5cm}
\caption{Mean absolute error and theoretical error bound
 in Table \ref{table1}.}
 \label{fig3}
\end{figure}

\begin{figure}[H]
\centerline{\includegraphics[width=0.96\columnwidth,draft=false]{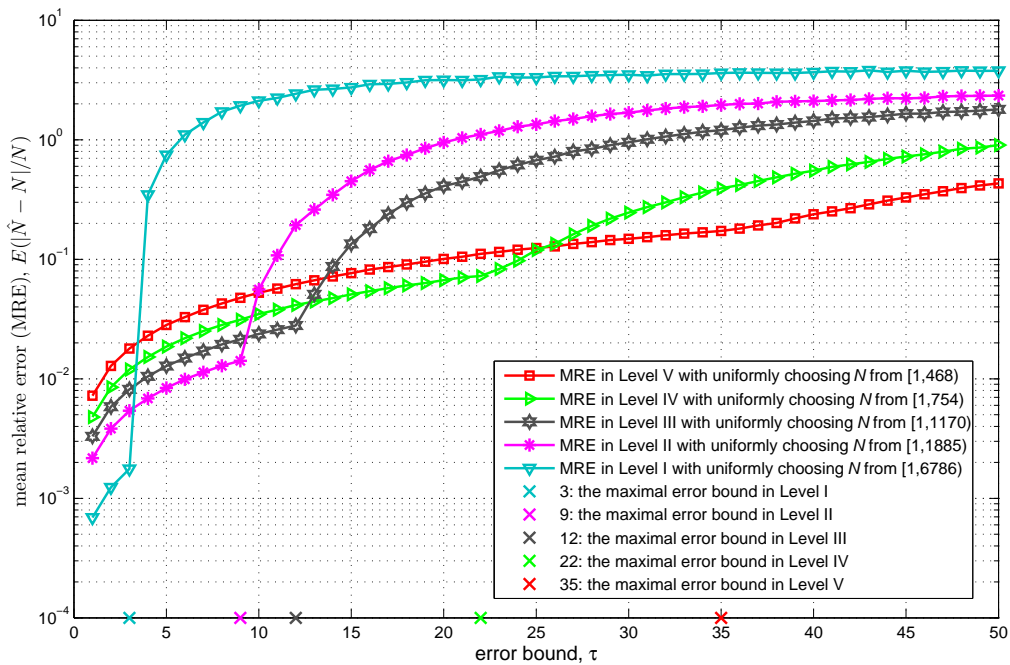}}
\vspace{-0.5cm}
\caption{Mean relative error and theoretical error bound in Table \ref{table1}.}
 \label{fig5}
\end{figure}

Next, for the same two moduli in Example \ref{ex2}, we check the correctness of the obtained dynamic range for Level \uppercase\expandafter{\romannumeral5}, Level \uppercase\expandafter{\romannumeral4}, Level \uppercase\expandafter{\romannumeral3}, Level \uppercase\expandafter{\romannumeral2} in Table \ref{table1}, respectively.
For each level, we let $N$ take some values close to the obtained dynamic range, as in Table \ref{my-label}, and
the remainder errors are uniformly distributed within the robustness bound. $2000000$ trials are implemented for each of $N$. It is shown in Table \ref{my-label} that when $N$ reaches the dynamic range, the mean absolute error $E(|\hat{N}-N|)$ begins to be much larger than the robustness bound, i.e., robustness does not hold.
\begin{table}[H]
\centering
\begin{tabular}{cc|cc|cc|cc}
\hline
\multicolumn{2}{c|}{Level \uppercase\expandafter{\romannumeral5}}                     & \multicolumn{2}{c|}{Level \uppercase\expandafter{\romannumeral4}}                    & \multicolumn{2}{c|}{Level \uppercase\expandafter{\romannumeral3}}                   & \multicolumn{2}{c}{Level \uppercase\expandafter{\romannumeral2}} \\ \hline
\multicolumn{1}{c}{$N$} & \multicolumn{1}{c|}{$E(|\hat{N}-N|)$} & \multicolumn{1}{c}{$N$} & \multicolumn{1}{c|}{$E(|\hat{N}-N|)$} & \multicolumn{1}{c}{$N$} & \multicolumn{1}{c|}{$E(|\hat{N}-N|)$} & \multicolumn{1}{c}{$N$} & \multicolumn{1}{c}{$E(|\hat{N}-N|)$} \\ \hline
465                     & 11.9995                       & 751                    & 7.6525                        & 1167                     & 4.3073                        & 1882                    & 3.3028                       \\
466                     & 11.9988                        & 752                    & 7.6496                        &1168                     &4.3080                        &1883                     & 3.2990                       \\
467                     & 11.9874                        & 753                     & 7.6573                        & 1169                     & 4.3072                        & 1884                    & 3.3012                       \\
\textbf{468}                     & \textbf{397.0580}                        & \textbf{754}                     & \textbf{675.4278}                        & \textbf{1170}                     & \textbf{1.0378e+03}                        & \textbf{1885}                    & \textbf{1.8344e+03}                       \\
469                     & 396.9849                       & 755                     & 675.4135                        & 1171                    & 1.0377e+03                       & 1886                     & 1.8345e+03                       \\
470                    & 397.0385                        & 756                    & 675.2734                        & 1172                     & 1.0374e+03                        & 1887                     & 1.8341e+03                       \\
\hline
\end{tabular}
\caption{Mean absolute error for some neighbors of the dynamic range in Table \ref{table1}.}
\label{my-label}
\end{table}

Finally, we compare the robustness among the single stage and two-stage robust CRT algorithms introduced in \cite{xiaoxia1}, and the proposed \textbf{Algorithm \ref{alg:Framwork22}} in this paper, for the case in Example \ref{ex3}. In the simulation, we consider the remainder error bound $\tau$ from $0$ to $25$, and the remainder errors are uniformly distributed on $[-\tau,\tau]$. The unknown integer $N$ is chosen uniformly at random and less than $13230$, and $2000000$ trials are implemented for each of $\tau$. The curves of the mean absolute error $E(|\hat{N}-N|)$ for the single stage robust CRT algorithm, the two-stage robust CRT algorithm, and \textbf{Algorithm \ref{alg:Framwork22}} in this paper are shown in Fig. \ref{fig4}. The simulation result demonstrates the improvement of the robustness bound for \textbf{Algorithm \ref{alg:Framwork22}}.

\begin{figure}[H]
\centerline{\includegraphics[width=1\columnwidth,draft=false]{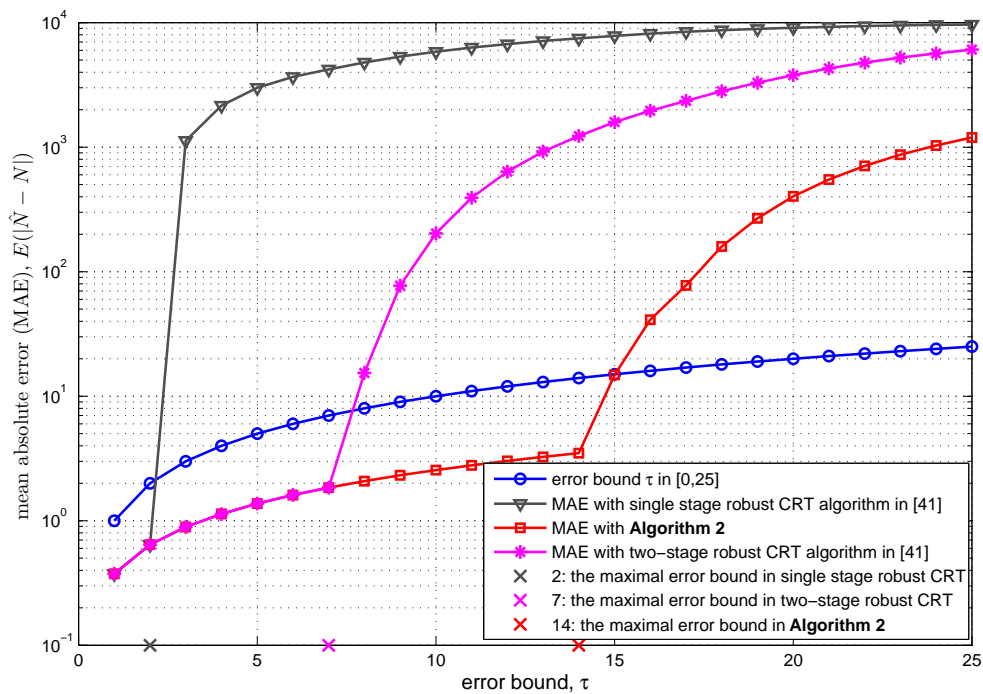}}
\vspace{-0.5cm}
\caption{Mean absolute error and theoretical error bound comparison using single stage and two-stage robust CRT algorithms, and \textbf{Algorithm \ref{alg:Framwork22}} in Example \ref{ex3}.}
 \label{fig4}
\end{figure}

\section{Conclusion}\label{sec7}
In this paper, we investigated a robust reconstruction problem of a large number from its erroneous remainders with several moduli, namely the robust remaindering problem. A relationship between the dynamic range and the robustness bound for two-modular systems was studied in this paper. Compared with the results in \cite{new}, we obtained a general condition on the remainder errors for the robustness to hold based on the idea of accurately determining the folding integers,
derived the exact dynamic range with a closed-form formula, and proposed simple closed-form reconstruction algorithms. We then considered the robust reconstruction for multi-modular systems by applying the newly obtained two-modular results, and generalized these two-modular results from integers to reals. We finally presented some simulations to verify our proposed theory.

\end{document}